%% file: DIRA.tex
\documentclass[aps,prl,superscriptaddress,twocolumn,amssymb,amsmath]{revtex4-1}
\usepackage{graphicx}
\usepackage{dcolumn}
\usepackage{bm}
\usepackage{braket}
\usepackage[dvipsnames]{xcolor}
\usepackage{comment}
\usepackage{tikz}
\usepackage{amsthm}
\usepackage{mathtools}
\usepackage{enumitem,kantlipsum}
\usepackage{soul}
\usetikzlibrary{backgrounds,fit,hobby,decorations,snakes,decorations.markings,positioning,shadows.blur,calc,calligraphy}
\tikzset{snake it/.style={decorate,decoration=snake,segment length=0.3cm}}

\usepackage{multirow} 

\DeclarePairedDelimiter\floor{\lfloor}{\rfloor}

\newtheorem{theorem}{Theorem}
\newtheorem{definition}{Definition}
\theoremstyle{remark}
\newtheorem*{remark}{Remark}

\definecolor{myred}{RGB}{236, 17, 0}
\definecolor{myblue}{RGB}{10, 88, 153}
\definecolor{mygreen}{RGB}{26, 152, 81}
\definecolor{myorange}{RGB}{236, 137, 0}
\definecolor{LightGray}{RGB}{220,220,220}
\definecolor{LinkColor}{RGB}{167,20,49}

\usepackage[colorlinks,citecolor=blue]{hyperref} 
\usepackage{mathptmx} 

\usepackage[framemethod=TikZ]{mdframed}
\mdfdefinestyle{MyFrame}{%
	linecolor=black,
	outerlinewidth=0.5pt,
	roundcorner=2pt,
	innertopmargin=10pt,
	innerbottommargin=10pt,
	innerrightmargin=10pt,
	innerleftmargin=10pt,
	backgroundcolor=gray!10!white,
}

\usepackage{float}
\usepackage{caption}

\begin{document}

\preprint{APS/123-QED}
\title{Device-Independent Randomness Amplification}

\author{Anatoly Kulikov}
\email{akulikov@phys.ethz.ch}
    \affiliation{Department of Physics, ETH Zurich, 8093 Zurich, Switzerland}
    \affiliation{Quantum Center, ETH Zurich, 8093 Zurich, Switzerland}

\author{Simon Storz}
    \affiliation{Department of Physics, ETH Zurich, 8093 Zurich, Switzerland}
    \affiliation{Quantum Center, ETH Zurich, 8093 Zurich, Switzerland}

\author{Josua D. Schär}
    \affiliation{Department of Physics, ETH Zurich, 8093 Zurich, Switzerland}
    \affiliation{Quantum Center, ETH Zurich, 8093 Zurich, Switzerland}

\author{Martin Sandfuchs}
    \affiliation{Department of Physics, ETH Zurich, 8093 Zurich, Switzerland}
    \affiliation{Quantum Center, ETH Zurich, 8093 Zurich, Switzerland}

\author{Ramona Wolf}
    \affiliation{Department of Physics, ETH Zurich, 8093 Zurich, Switzerland}
    \affiliation{Quantum Center, ETH Zurich, 8093 Zurich, Switzerland}
    \affiliation{Naturwissenschaftlich-Technische Fakultät, Universität Siegen, 57068 Siegen, Germany}

\author{Florence Berterotti\`{e}re}
    \affiliation{Department of Physics, ETH Zurich, 8093 Zurich, Switzerland}
    \affiliation{Quantum Center, ETH Zurich, 8093 Zurich, Switzerland}

\author{Christoph Hellings}
    \affiliation{Department of Physics, ETH Zurich, 8093 Zurich, Switzerland}
    \affiliation{Quantum Center, ETH Zurich, 8093 Zurich, Switzerland}

\author{Renato Renner}
    \affiliation{Department of Physics, ETH Zurich, 8093 Zurich, Switzerland}
    \affiliation{Quantum Center, ETH Zurich, 8093 Zurich, Switzerland}

\author{Andreas Wallraff}
    \affiliation{Department of Physics, ETH Zurich, 8093 Zurich, Switzerland}
    \affiliation{Quantum Center, ETH Zurich, 8093 Zurich, Switzerland}

\date{\today}

\begin{abstract}

Successful realization of Bell tests \cite{Bell1964,CHSH1969,Freedman1972,Aspect1982a,Rowe2001,Salart2008} has settled an 80-year-long debate \cite{Einstein1935,Bell2004}, proving the existence of correlations which cannot be explained by a local realistic model. Recent experimental progress allowed to rule out any possible loopholes in these tests \cite{Hensen2015,Giustina2015a,Shalm2015}, and opened up the possibility of applications in cryptography envisaged more than three decades ago \cite{Ekert1991}. A prominent example of such an application is device-independent quantum key distribution, which has recently been demonstrated \cite{Nadlinger2022,Zhang2022l,WenZhao2022}. One remaining gap in all existing experiments, however, is that access to perfect randomness is assumed \cite{Bierhorst2018,Liu2018,Zhang2022l,Nadlinger2022,Storz2024}. To tackle this problem, the concept of randomness amplification has been introduced \cite{Colbeck2012}, allowing to generate such randomness from a weak source---a task impossible in classical physics. In this work, we demonstrate the amplification of imperfect randomness coming from a physical source. It is achieved by building on two recent developments: The first is a theoretical protocol implementing the concept of randomness amplification within an experimentally realistic setup \cite{Kessler2020}, which however requires a combination of the degree of Bell inequality violation (S-value) and the amount of data not attained previously. The second is experimental progress enabling the execution of a loophole-free Bell test with superconducting circuits \cite{Storz2023}, which offers a platform to reach the necessary combination. Our experiment marks an important step in achieving the theoretical physical limits of privacy and randomness generation \cite{Ekert2014}.

\end{abstract}

\pacs{Valid PACS appear here}

\maketitle

Randomness is a term that immediately conjures associations, as it is an idea deeply ingrained in our understanding. Yet, in a scientific context, we necessitate a precise definition to assess whether a protocol indeed produces randomness. In this context, randomness is defined as being \emph{fundamentally unpredictable}, which means that the laws of physics forbid the prediction of its values. This property extends beyond statistical correlations in the bit string---it is an inherent characteristic of the process employed to generate randomness \cite{Frauchiger2013}. For instance, consider a book filled with numbers generated by a high-quality random number generator \cite{Randomnumbers}. While these numbers may exhibit strong statistical properties (e.g., not containing any patterns), they are unsuitable for use in cryptography, as anyone with access to the book could easily break the encryption. Consequently, certifying randomness always involves certifying the generation process.

While the assumption of access to perfect randomness is prevalent in many tasks, practical availability raises questions. This issue is exemplified by investigations into public RSA \cite{Rivest1978} keys: due to inadequate randomness, a non-negligible proportion of publicly available keys on the internet exhibit common factors, rendering them susceptible to efficient factorization, thus undermining the security of RSA \cite{Heninger2012,Lenstra2012}.

Conventional random number generators, rooted in classical physical processes, grapple with a foundational concern---the potential for adversaries to predict their outputs by scrutinizing the microscopic degrees of freedom, thereby eroding their essential unpredictability.

Quantum-mechanical processes, on the other hand, feature innate randomness and therefore offer a natural ground to build such devices. With the growing access to single and controllable quantum systems on a variety of platforms and prospects of perfectly random measurement outcomes, a plethora of quantum random number generators (QRNGs) were proposed and experimentally demonstrated \cite{Herrero-Collantes2017}, some of which are now commercially available. 
However, this approach is not immune to attacks that exploit flawed implementations, including issues like noise, detector dead times, and dark counts. Whether these design flaws are unintentional---arising from the limits of our engineering capabilities---or deliberately introduced by a powerful adversary, they compromise the randomness produced by the setup \cite{Dhara2014, Hurley2020}.

To overcome these issues, methods have been introduced \cite{Mayers2003,Colbeck2009,Pironio2010,Christensen2013c,Bierhorst2018,Liu2018} which certify the random nature of a bit string based on the quantum mechanical feature of non-local correlations, demonstrated through the violation of Bell-like inequalities~\cite{Bell1964,CHSH1969,Freedman1972,Aspect1982a,Rowe2001,Salart2008}. This allows certain protocols to be \textit{device-independent}, i.e., explicitly discarding the requirements of the exact knowledge of the employed device, and assumptions of its strict correspondence to a theoretical model. These protocols therefore enable certifying the randomness of a produced output while even allowing the employed system to be malicious.

A discerning reader may observe what seems to be a circular argument here: Executing a Bell test already necessitates perfect input randomness to make choices among different measurement settings~\cite{Bierhorst2018,Liu2018,Zhang2022l,Nadlinger2022,Storz2024}. Being precise, what can be accomplished with this approach is the \emph{expansion} of randomness---more randomness is obtained than initially used \cite{Shalm2021,Liu2021}. However, this does not solve the initial issue of generating such randomness in the first place. Fortunately, quantum mechanics offers an even more intriguing possibility: perfect randomness can be created even if we only require that the inputs for a protocol exploiting non-locality are not entirely deterministic. This is known as \emph{randomness amplification} \cite{Colbeck2012}.

Naturally, achieving perfect randomness from such weak assumptions does come with its challenges. Initial theoretical proposals featured impractical properties, such as a large required number of devices, no noise tolerance, or a very small amount of output bits \cite{Colbeck2012,Gallego2013,brandao2016}. Furthermore, to ensure the validity of such a conclusion solely based on fundamental principles, in a practical setting a set of \textit{loopholes} has to be closed in the experimental implementation. We discuss this aspect in the Methods. Experimentally, even executing a Bell test free of loopholes, a milestone that was only reached in the past decade \cite{Hensen2015,Shalm2015,Giustina2015a}, is insufficient: 
non-vanishing improvement in randomness quality requires a combination of a Bell inequality violation and an amount of generated data which was not achieved in previously reported works. Given these challenges, randomness amplification remains an unattained experimental goal---until now.

\begin{figure}
  \begin{center}
  \input{figures/setup}
  \caption{\textbf{Schematics of the randomness amplification protocol.} Two public imperfect SV sources are used to provide input randomness (strings $\mathbf{X}$ and $\mathbf{Y}$) to a two-node untrusted black-box device $\lambda$. Performing verifiably non-local measurements, we obtain output strings \textbf{A} and \textbf{B}. We do not make any assumptions about the inner workings of the device $\lambda$ and attribute all possible side-information to the adversary.
  If the output strings pass a pre-defined quality threshold (see main text for the description), we run a randomness extractor with an extra input (string \textbf{Z}) from the public SV source to generate the amplified private random string $\mathbf{K}$.
  }
  \label{fig:protocol}
  \end{center}
\end{figure}
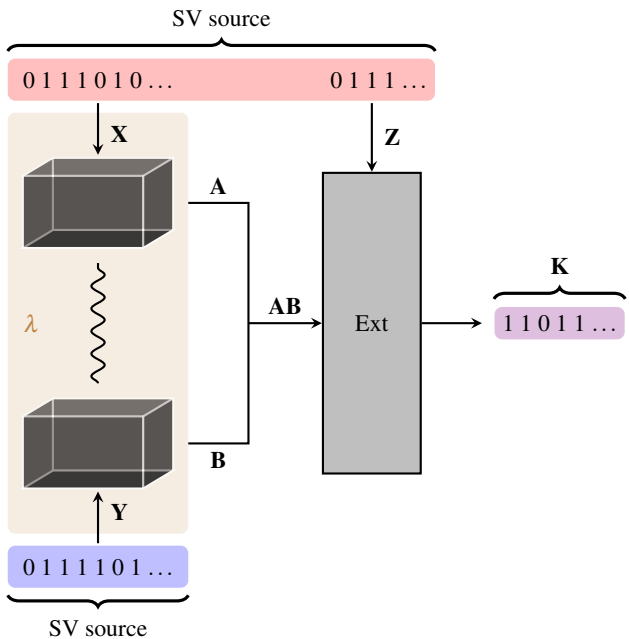


In this work, we report an experimental realization of randomness amplification. Specifically, using the resource of non-local correlations certified by a loophole-free violation of a Bell inequality, and a public, imperfect source of randomness, we obtain near-perfect private randomness as an output. More precisely, two parties A and B with access to a biased public Santha-Vazirani (SV) source (see a formal definition below) perform Bell-type correlation measurements on their quantum systems to amplify the randomness of the initial seed using a classical extractor, see Fig.~\ref{fig:protocol}.

We base the implementation on a solid-state circuit QED platform~\cite{Blais2021} and a similar setup as the recently reported loophole-free Bell test \cite{Storz2023}. The main ingredient of the experimental system---a two-node untrusted device $\lambda$---is realised by two dilution refrigerators connected via a modular cryogenic link spanning a distance of 30 meters, see Ref.~\cite{Schaer2023}. We encode quantum information in two transmon qubits fabricated on separate, nominally identical quantum devices, operated in the cryostats at sites A and B at temperatures of about 15~mK. In addition to control and readout circuitry, each qubit is locally coupled to an on-chip transfer resonator and a Purcell filter. For the generation of remote entanglement, we connect the transfer resonator circuits of the devices at both sites through a rectangular aluminium waveguide \cite{Kurpiers2017,Magnard2020}. 
The two-node device $\lambda$ is further described in Supplementary Information section I. 

We access imperfect random bits from quantum random number generators introduced in Ref.~\cite{Abellan2015}, which provide on-demand random bits based on the process of laser phase diffusion. Single-shot readout of the state of the quantum devices generates the protocol output strings $\mathbf{A}$ and $\mathbf{B}$. Both the QRNGs and the analog-to-digital converters (ADCs) recording the measurement results form part of the trusted part of the setup (and are therefore outside of $\lambda$).

In the following, we first introduce the theoretical protocol for device-independent randomness amplification with the requirements for ``black box'' devices (i.e., an abstract interface), and then detail the experiment based on our hardware realizing it.


Relaxing the requirement for perfect input randomness naturally requires a measure of ``imperfection''. Imagine an observer, also called an adversary, who attempts to guess the output of a source producing a single bit. She would assign to each value $b \in \{0,1\}$ some probability $\mathrm{Pr}[b|e]$, where $e$ denotes any information which she can use to help her predict the value of the bit. More generally, we consider a sequence of bits $b_{1}, \ldots, b_{n}$. After having observed the first $i-1$ bits, the adversary's knowledge is characterized by $\mathrm{Pr}[b_{i}|b_{1}, \ldots,  b_{i-1},e]$.
The Santha-Vazirani (SV) model then considers sources for which none of the bits are completely determined, i.e., $\frac{1}{2} - \mu \leq \mathrm{Pr}[b_{i}|b_{1}, \ldots, b_{i-1},e] \leq \frac{1}{2} + \mu$ holds for all $i$ and $b_{1}, \ldots, b_{i}, e$. The parameter $\mu$ is called the bias of the source and $\mu = 0$ corresponds to a sequence of uniformly random bits. In our setup we do not have a single source but two sources, each producing a sequence of bitstrings. We then require the above condition to be satisfied for the bitstring containing the output of both sources (in the order in which they were produced). In particular, this means that the two sources may be correlated, as long as the condition above is satisfied.

The device-independent randomness amplification protocol we have implemented inherits from the recent proposal in Ref.~\cite{Kessler2020}. It requires a two-node system capable of violating a Bell-like measurement dependent locality (MDL) inequality and two imperfect, possibly correlated, public sources of randomness as per Fig.~\ref{fig:protocol}.
The MDL inequality can be thought of as a special type of Bell inequality which incorporates the input imperfections (see Methods). Following the device-independent framework, where no assumptions are made about the physical device, we outline the interface requirements of the protocol in Box I.

\begin{table}
	\captionsetup{name=Box}
	\begin{mdframed}[style=MyFrame]
		\small
		{\large \textsf{Box I: Protocol (abstract)}}\\[0.2cm]
		The protocol is defined in terms of the following parameters:\\[0.1cm]
		\begin{tabular}{rl}
			$n\in\mathbb{N}$: & number of protocol trials\\
			$\mu$: & bias of the two $\mu-$SV sources\\
			$m\in\mathbb{N}$: & length of the final random bit string \\
            $\mathrm{Ext}$:& $\{0,1\}^{2n} \times \{0,1\}^{d} \rightarrow \{0,1\}^{m}$ \\
            & the two-source randomness extractor.
		\end{tabular}\\[0.1cm]
		\hrule\vspace{2pt}
		\begin{enumerate}[leftmargin=*,topsep=0pt]
			\item \textbf{Measurement:} For each trial $i\in\{1,2,\dots,n\}$:
			\begin{enumerate}[noitemsep,topsep=0pt]
				\item Sample $x_i, y_i \in \{0, 1\}$ from the two $\mu-$SV sources
				\item Feed the devices with the inputs $x_i$, $y_i$ and record the respective outputs $a_i, b_{i} \in\{0,1\}$.
			\end{enumerate}
			\item \textbf{Parameter estimation:} Use the inputs $\mathbf{X}$, $\mathbf{Y}$ and the outputs $\mathbf{A}$, $\mathbf{B}$ to evaluate the observed value of the MDL inequality $S_{\mu,\mathrm{obs}}$. If $S_{\mu,\mathrm{obs}} < 0$, the protocol aborts and no random bits can be produced.
			\item \textbf{Randomness extraction:}
			    \begin{enumerate}[noitemsep,topsep=0pt]
			        \item Draw another bit string $\mathbf{Z} \in \{0, 1\}^{d}$ from one of the sources.
                    \item Apply the extractor to produce the final output $\mathbf{K}=\mathrm{Ext}(\mathbf{AB}, \mathbf{Z}$) of length $m$.
			    \end{enumerate}
		\end{enumerate}
	\end{mdframed}

    \vspace*{-1.7\baselineskip}
\end{table}

The protocol begins by running a Bell-type experiment where the amount of produced randomness is quantified by the observed MDL violation $S_{\mu,\mathrm{obs}}$ \cite{Kessler2020}. Applying a randomness extractor to the output of the device $\lambda$ and extra seed $\mathbf{Z}$ from the public SV source then produces the final bit string $\mathbf{K}$ which can be made arbitrarily close to perfectly random (characterized by a security parameter $\varepsilon$). To guarantee the quality of the bit string $\mathbf{K}$, the parameters in Box I need to be chosen appropriately. The full details of how to pick these parameters are discussed in Supplementary Information section VI.

\begin{figure}
    \centering
    \includegraphics[width=\columnwidth]{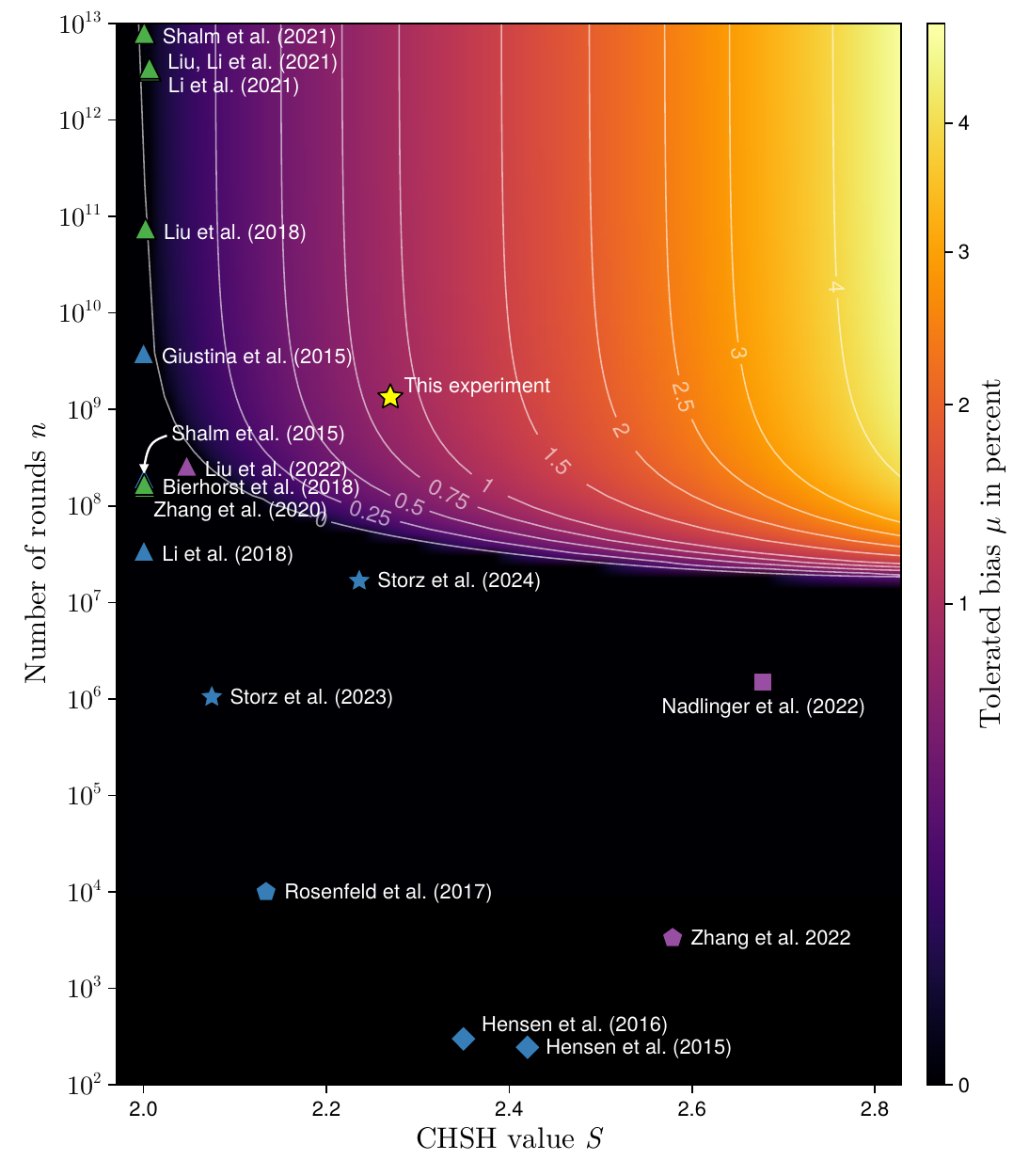}
    \caption{\textbf{Performance of the employed device-independent randomness amplification protocol under different experimental conditions.} The plot shows the trade-off between CHSH violation $S$ and the number of rounds $n$ required to achieve at least a million bits of output. The noise model assumed here is depolarizing noise. Note that this plot is subject to future theoretical advancements; for instance, improved extractors could allow to amplify even sources with higher bias $\mu > 4\%$. We further list previously published experiments. The marker shape represents the experimental platform: Photons (triangles), NV-centers (rhombi), neutral atoms (squares), trapped ions (pentagon) and superconducting circuits (stars). The color represents the type of experiment: Device-independent randomness expansion (green), loophole-free Bell test (blue), device-independent quantum key distribution (purple) and device-independent randomness amplification (yellow). The listed QKD experiments do not close the locality loophole.}
    \label{fig:S_vs_n_theory}
\end{figure}

\begin{figure*}
  \begin{center}
  \includegraphics[width=2\columnwidth]{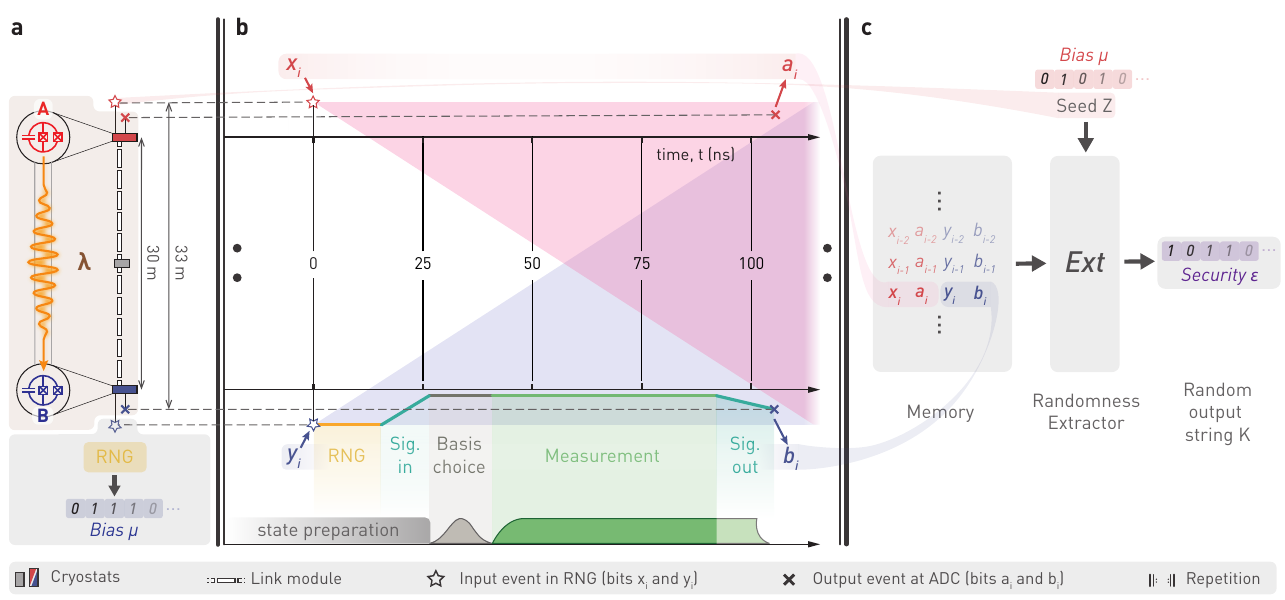 }
  \caption{\textbf{Full experimental protocol.} \textbf{a}, Experimental setup. Two dilution cryostats, separated by 30 meters, contain a superconducting qubit each, forming the nodes A and B. In combination, they represent the untrusted device $\lambda$. Stars mark the location of input QRNGs, and crosses mark signal acquisition ADCs. The distance between the starting event and the end of the acquisition is therefore 33 meters. The inset on the bottom left corner highlights one of the two RNGs, producing a string of random bits with bias $\mu$. \textbf{b}, Space-time diagram of an experimental trial \textit{i}. Each trial starts from sampling a random bit ($x_i$, (red) for Alice, and $y_i$, (blue) for Bob) from the QRNGs and ends with a discretized measurement result ($a_i, b_i$). Individual experimental segments are indicated below  by different colours: the RNG random number production time, the duration of the basis choice pulse and the measurement, and the duration of the signals entering and leaving the cryostats. \textbf{c}, Schematic representation of the post-processing routine, yielding a certified private random string with security $\epsilon$, see text for details.}
  \label{fig:protocol_experiment}
  \end{center}
\end{figure*}

We note that the length $m$ of the private random output bit string produced by the extractor $\mathrm{Ext}$ is determined by the tolerated source bias $\mu$, the number of trials $n$ and the observed MDL violation $S_{\mu,\mathrm{obs}}$, see Figure~\ref{fig:S_vs_n_theory}.
Specifically, in this experiment we aim to extract a string of the size exceeding 1~Mbit, partially motivated by this length being a requirement of the input randomness to recently reported device-independent quantum key distribution (DI-QKD) routines (more on this point in the outlook discussion). To achieve this, considering the achieved quality of the experiment (corresponding to $S \approx 2.27$), the required security parameter $\varepsilon$ and the specific extractor employed (see Methods), we require on the order of $n=10^9$ trials in the experiment, entering a previously unmatched experimental regime (Figure~\ref{fig:S_vs_n_theory}). This volume of trials ensures that the finite size of the recorded data set does not substantially lower the tolerance of our protocol against the input bias $\mu$.


Our experimental implementation is shown in Fig.~\ref{fig:protocol_experiment} and follows the abstract interface in Box I.  
Two third-party commercial QRNGs \cite{Abellan2015} serve as imperfect public sources of randomness. In each trial $i$, after performing active all-microwave reset to bring both qubits into the ground state \cite{Magnard2018}, we prepare an entangled $\ket{\Psi^+}$-Bell state using a deterministic remote entanglement generation scheme \cite{Cirac1997,Kurpiers2018}, based on the exchange of a single microwave photon between the remotely connected qubits. Sampling a random input bit $x_i$ ($y_i$) at node A (B) from the public $\mu-$SV source on each site marks the beginning of the time-critical part of the randomness amplification protocol. At each node, as in Ref.~\cite{Storz2023}, the random input bit controls the state of a microwave switch which conditionally passes a $\pi/2$ microwave rotation pulse to the qubit, implementing the measurement basis selection. Immediately afterwards, we start the measurement of the state of the local qubit by sending a microwave pulse to the quantum device, interacting with the transmon qubit~\cite{Walter2017}. We then amplify the readout signal and route it out of the cryostat to digitize it at room temperature using an analog-to-digital converter (ADC) and post-process it using an FPGA. We consider the measurement result to be fixed once the last part of the readout signal reaches the input of the ADC, employing similar considerations as discussed in Ref.~\cite{Storz2023}. With this scheme, we achieve single-shot readout fidelity of about 98$\%$ in only 50~ns for both qubits.

After the pre-defined number of trials $n$, we collect the input-output statistics of each round ($x_i,y_i,a_i,b_i$) to calculate the MDL and Clauser–Horne–Shimony–Holt (CHSH) inequalities, see Methods. Taking into account the data of each trial closes the fair-sampling loophole \cite{Garg1987, Eberhard1993}, and not assuming the individual trials are independent and identically distributed closes the memory loophole \cite{Larsson2014}.


We perform randomness amplification based on $n=20 \cdot 2^{26} = 1,342,177,280$ trials acquired over a total time of about 9~hours.
We split the data acquisition in 20 blocks of $2^{26}$ ($\sim 67$~million) trials each. Trials in each block are run back-to-back with the repetition time of $20~\mu$s, i.e. a rate of $50$~kHz. After acquiring a data block, we re-calibrate the readout acquisition weights \cite{Walter2017} and perform the verification of the system's timings. Furthermore, as closing the locality loophole is of paramount importance for the device-independent nature of the protocol, we implement a real-time monitoring of the setup synchronisation during the device-independent randomness amplification execution. Specifically, we introduce a series of oscilloscope measurements observing the arrival time of auxiliary pulses from various devices in the setup, hence ensuring the absence of timing jitter of control pulses, triggers and readout signal. In Supplementary Information sections II and III we discuss the triggering and timing verification schemes and the system-wide timing stability data during the protocol execution.

\begin{figure}
  \begin{center}
    \includegraphics[width=\columnwidth]{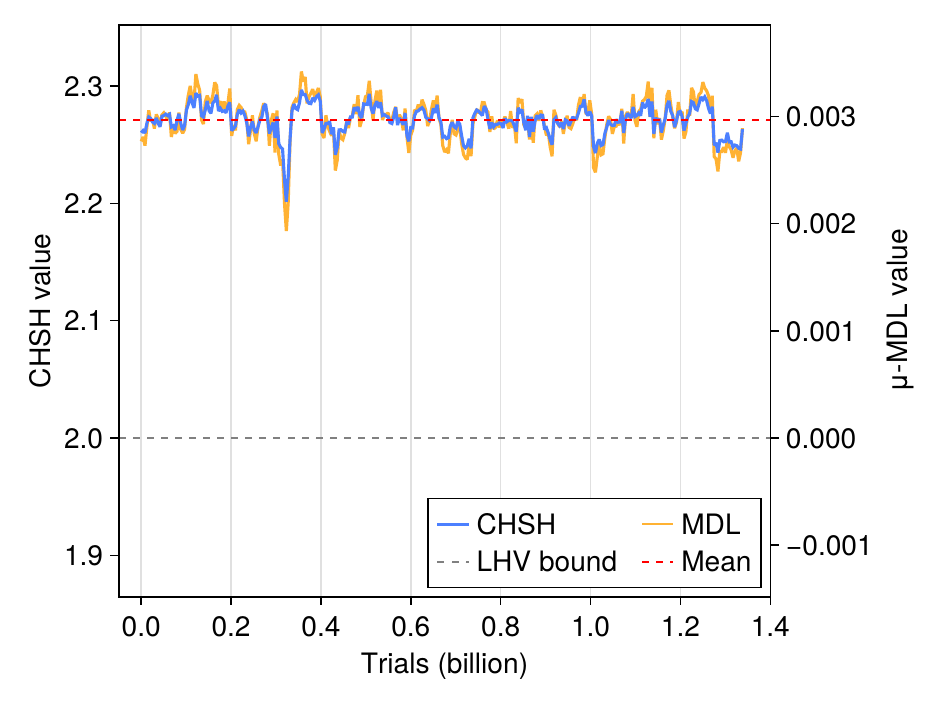}
    \caption{\textbf{S-value \emph{vs.} iteration.} Quality of the data and stability visualization. Each point corresponds to $\approx$ 4 mil ($2^{22}$) experimental trials. The red horizontal line indicates the average over the entire dataset, the gray horizontal line represents the classical bound of $S=2$ or $S_{\mu}=0$.
  }
  \label{fig:S_vs_time}
  \end{center}
\end{figure}

Before the protocol begins, we must pick a tolerated bias $\mu$. The largest $\mu$ that one can pick depends on the expected MDL inequality violation $S_{\mu,\mathrm{obs}}$ which in turn depends on the specific noise in the experiment (see Methods).
For our experiment we see that we can pick $\mu$ to be slightly above $0.75\%$ (see Fig.~\ref{fig:S_vs_n_theory}) and therefore we pick $\mu = 0.75\%$; For comparison, the excess predictability specified by the manufacturer of our QRNGs is below $\varepsilon < 1\cdot 10^{-5}$ \cite{Abellan2015}. In our experiment, we measure an average CHSH value of $S_{\mathrm{obs}}=2.271$ and an average $\mu$-MDL value of $S_{\mu,\mathrm{obs}} = 0.00296$, see Fig.~\ref{fig:S_vs_time} for the dependence of the $\mu$-MDL and CHSH inequality violations over time.
As all the loopholes are closed, and the value $S_{\mu,\mathrm{obs}}$ exceeds the MDL inequality threshold $S_{\mu, \mathrm{crit.}}=0$, the obtained dataset meets the formal requirements of the abstract device-independent randomness amplification protocol (see Box I). We are hence able to proceed to randomness extraction.

Even though the outputs $\mathbf{AB}$ of the Bell-type experiment possess an inherent amount of randomness \cite{Ekert1991, Acin2007, Pironio2010}, they may be strongly correlated with each other or with side-information held by an adversary. Consequently, they do not satisfy the definition of $\varepsilon$-randomness which loosely speaking requires the outputs to be distinguishable from perfectly uniform with probability at most $\varepsilon$ (see Methods). To break these correlations, we apply a randomness extractor which processes two sources of imperfect randomness into a new string of nearly perfect randomness. The technical requirements for the extractor to work are outlined in the Methods. 

Since the original work of Santha and Vazirani \cite{Santha1986}, randomness extractors have been the focus of intense study in classical information theory (see \cite{Shaltiel2011} for a review). For example, it has been shown that any two-source extractor which is secure against classical side-information is also secure against quantum side-information given that the size of the output is sufficiently reduced \cite{Arnon-Friedman2016}. This means that, at least in principle, we could use any of the classical two-source extractors. However, constructions differ significantly in both their performance (i.e. the requirements on the quality of input randomness) and the computational complexity. As a large amount of data is required to produce a Mbit-range output bit string, we require our extractor to be computable in quasi-linear memory and time. For this reason, we choose the two-source extractor from \cite{Vazirani1987},
which is sub-optimal in terms of requiring a large amount of entropy in the two input sources, but can be implemented with computational cost of $O(n\log n)$ \cite{Foreman2023}.

As a final result of the device-independent randomness amplification protocol, using the chosen extractor, we generate a private $\varepsilon$-random random bit-string $\mathbf{K}$ consisting of 20,431,465 bits. The security parameter $\varepsilon$ defining the trace distance between the generated randomness and perfect randomness was chosen to be $\varepsilon = 10^{-12}$. That is, except with a negligible failure probability of $10^{-12}$, we have extracted a uniformly random bitstring from the SV-source.
The generated string $\mathbf{K}$ is private and is guaranteed to be $\varepsilon$-random provided the bias $\mu$ of the input source satisfies $\mu \leq 0.75\%$.

While it is not possible to infer anything about the process used to generate the randomness from the output bit string alone, we still find it meaningful to run statistical tests on the final string produced by the protocol and ensure that the result is consistent with conventional state-of-the-art random number generators.
To verify the quality in the statistical manner, we run the National Institute of Standards and Technology (NIST) statistical test suite \cite{Rukhin2001} and Diehard~\cite{Marsaglia2002} batteries of statistical tests 
on the string $\mathbf{K}$ generated as a result of the full execution of device-independent randomness amplification protocol. We find that the result passes all statistical tests in the NIST suite and all tests from the Diehard battery for which the string is sufficiently long.

In summary, we have demonstrated experimentally that it is possible to amplify the quality of physical randomness. More precisely, starting from $5,368,709,120$ low-quality random bits with a bias of up to $0.75 \%$, we produced $20,431,465$ bits of almost perfect randomness, guaranteed to be equal to a uniformly distributed bit string with probability at least $1-\varepsilon = 1- 10^{-12}$. This guarantee holds in a device-independent manner, meaning that no assumption has to be made about the quantum devices, except for the shielding of their outputs from adversaries. For the precise assumptions on the classical components, e.g., the ones used for post-processing, see  Supplementary Information section~V.

The device-independent amplification of randomness is highly demanding: it requires an experimental setup where a large quantity of high-quality entangled qubits are distributed and measured at space-like separation. In particular, the requirements on the measured correlations are substantially higher than for loophole-free Bell tests, as shown in Fig.~\ref{fig:S_vs_n_theory}. This was achieved with fast measurement, data acquisition and processing time ($<22~\mu$s/trial in total), multi-day room-temperature control electronics and calibration stability, enabling us to collect Gbit-range amount of data, and below $14\%$ total loss connection between the quantum bits at nodes, enabling us to distribute high-fidelity entangled states and reach S-values of $S=2.27$. In the Methods section we outline key improvements of the setup compared to the prior works.

The amplified randomness may be used directly as a public source of randomness, comparable to the NIST randomness beacon, see \cite{Kelsey2019}. The advantage of the device-independent randomness amplification scheme is that the random bits come with a certificate (the data showing the Bell violation), which proves that they have full entropy, even conditioned on the low-quality input bits. Note that the bit rate of our scheme is comparable to that of the current NIST randomness beacon.

High-quality randomness is necessary to guarantee the security of cryptographic schemes. For example, the claim that quantum key distribution is secure relies on the assumption that the random bits that are used in the protocol, e.g., for choosing the measurement bases, have full entropy. Similarly, cryptographic tasks such as bit commitment, encryption, secret sharing, zero-knowledge proofs, and secure multi-party computation cannot be implemented securely with sources of less-than-perfect randomness \cite{DodisShien2004}. 
This is particularly crucial for schemes that aim at the ultimate physical limits of security, such as device-independent quantum key distribution~\cite{Ekert2014}. Assuming the availability of perfect randomness would be incompatible with this goal, hence the use of device-independent randomness amplification seems unavoidable. Due to the high repetition rate of our experiment, the amount of randomness produced in this work exceeds the randomness consumed by the recently published DI-QKD implementations \cite{Nadlinger2022,Zhang2022l}.

In the future, when multi-node networks with high-fidelity communication links are available, device-independent randomness amplification may be combined with distributed computing protocols \cite{Yao1982,Goldreich1987}. This would reduce the assumptions on the classical processing. The quality of the output randomness could then be guaranteed even if a certain number of network nodes are corrupted. 

\section*{Acknowledgments}
We thank Raul Conchello Vendrell for contributions to the software framework, Reto Schlatter for support with the operation of the cryogenic setup, Anna Efimova for assistance with the synchronization of the electronic instruments, and Nicolas Gisin and Matteo Fadel for comments on the manuscript.

\paragraph*{Funding:}
The work was funded by the European Union’s Horizon
2020 FET-Open project SuperQuLAN (899354), the Air Force Office of Scientific Research (AFOSR), grant No. FA9550-19-1-0202, the
QuantERA project eDICT, the National Centre of Competence in Research SwissMAP, the Quantum Center at ETH Zurich and by ETH Zurich. RW acknowledges support from the Ministry of Culture and Science of
North Rhine-Westphalia via the NRW-Rückkehrprogramm.

\section*{Author contributions}
A.K., S.S. and J.S. assembled the experimental setup, performed the measurements and analysed the data. A.K., S.S., F.B., C.H. and J.S. developed the control software. R.W. and M.S. performed the theoretical analysis, randomness extraction and developed the security proofs.
A.K., S.S., R.W., M.S., J.S. and R.R. wrote the manuscript with input from all authors. R.R. and A.W. supervised the project.

\section*{Competing Interests}
The authors declare no competing interests.

\section*{Data Availability}
All data is available from the corresponding authors upon reasonable request.

\section*{METHODS} 

{\it Definition of randomness --} The definition of randomness used in this work follows \cite{Colbeck2012} and \cite{Frauchiger2013}. Here, the notion of a \emph{space-time variable} is introduced. This is a random variable with an associated coordinate that indicates the physical location of the value in relativistic space-time. The output $X$ of a random process as well as any side information is modeled by such space-time variables. Let $X$ denote the space-time variable whose coordinate indicates the point in space-time where the process of generating $X$ is started, also called the \emph{trigger event}. $X$ is called $\varepsilon$-random if it is $\varepsilon$-close to uniform and uncorrelated to all other space-time variables that are outside of the future light cone of $X$. Denoting this set by $\Gamma_X$, this can be expressed as
	\begin{equation}
		\label{eq:randomprob}
		\frac{1}{2}\|P_{X\Gamma_X}-U_{X}\times P_{\Gamma_X}\|_1\le\varepsilon,
	\end{equation}
where $U_X$ is the uniform distribution, i.e., $U_{X}(x)=\frac{1}{|\mathcal{X}|}$ for all $x\in\mathcal{X}$, and 
	\begin{equation}
		\frac{1}{2}\|P_X-Q_X\|_1=\frac{1}{2}\sum_x|P_X(x)-Q_X(x)|
	\end{equation}
is the trace distance between two probability distributions.

The independence condition in (\ref{eq:randomprob}) can be lifted to a condition on quantum states:
	\begin{equation}
		\label{eq:randomquantum}
		\frac{1}{2}\|\rho_{XE}-\omega_X\otimes\rho_E\|_1\le\varepsilon,
	\end{equation}
where $\|\cdot\|_1=\mathrm{tr}(|\cdot|)$ denotes the trace norm and $\omega_X=\frac{1}{|\mathcal{X}|}\mathrm{id}_X$ denotes the maximally mixed state on $X$. This condition describes the states for which $X$ is $\varepsilon$-close to uniformly random and independent of the quantum side information $E$. Importantly, the trace norm can only decrease when we apply a completely positive trace preserving map (e.g., a measurement) on the system $E$. Hence, the independence condition in (\ref{eq:randomprob}) follows from the quantum version (\ref{eq:randomquantum}). 

The condition in (\ref{eq:randomquantum}) ensures that the random variable $X$ is uncorrelated with anything at an earlier time. Thus, in line with what was mentioned in the introduction, the value of $X$ is unpredictable, guaranteed by the laws of physics.

From the definition of the employed two-source randomness extractor \cite{Dodis2004}, we see that proving that condition (\ref{eq:randomquantum}) holds reduces to finding a lower bound on the min-entropy. This quantity, in turn, is related to the MDL inequality violation in \cite{Kessler2020}. As such, an observation of the MDL inequality violation directly allows to prove the independence condition (\ref{eq:randomquantum}), and hence certifies $\varepsilon$-perfect randomness.

\

{\it Loopholes and assumptions in experimental implementations --} As device-independent randomness amplification aims to certify the output of a device solely based on fundamental physics principles, it is critical to ensure that an experimental implementation makes no additional assumptions about the setup. If such assumptions are present, they open so-called \textit{loopholes} in the conclusions drawn from the experiments \cite{Larsson2014}, similarly to the situation with Bell tests.
Specifically, the \textit{locality loophole} \cite{Brunner2014} opens when the two nodes A and B are able to exchange classical information throughout each trial of the experiment. This would allow them to adapt the measurement result on one node based on the measurement basis choice on the other node, and thereby to achieve strong correlations in the input-output-bit statistics that are however not necessarily non-local. This issue can be addressed by setting up the experiment in a space-time configuration, where the two nodes are space-like separated throughout each trial of the protocol. This ensures that in each trial, no signal traveling at most at the speed of light, possibly carrying information about the measurement basis of one node, can reach the other node before the latter has completed its measurement. Therefore, carefully determining the space-time configuration of the start and stop events of each trial ensures that the locality loophole is closed.
Furthermore, the \textit{fair-sampling loophole} \cite{Garg1987,Eberhard1993} opens if assumptions are made about the distribution of the underlying input-output-bit correlations, for instance because one does not have access to the measurement results of all trials. This loophole is closed by including the data captured in all trials of the experiment.
Third, the \textit{memory loophole} \cite{Larsson2014} is addressed by not assuming that the results of each trial are independent and identically distributed.

A Bell test closing all these loopholes simultaneously is typically referred to as \textit{loophole-free}, and protocols based on Bell tests that achieve the same feat are termed \textit{device-independent}, as they do not rely on assumptions about the employed devices.
\

{\it Entropy production; measurement dependent locality (MDL) inequality --}
The security of our device-independent randomness amplification protocol hinges on the property that the measurement results of Bell tests exhibit a certain degree of inherent randomness \cite{Ekert1991, Acin2007, Pironio2010}. For most Bell inequalities, including the widely known CHSH inequality, this property only holds under the assumption that the inputs are chosen uniformly at random. This assumption is clearly not met in a device-independent randomness amplification experiment where the goal is precisely to process imperfect randomness into better randomness.
Therefore, what we require for device-independent randomness amplification is a special type of Bell inequality which can be used to certify the randomness of the devices' outputs even when the inputs are chosen using a SV source. This is exactly what is achieved by the following Bell inequality (also called a measurement dependent locality or MDL inequality) \cite{Putz2014a}.
\begin{equation}
\begin{aligned}
    \label{eq:mdl_inequality}
    0 \geq S_\mu \equiv &\mu_{\mathrm{min}}P_{ABXY}(0000) \\
    &- \mu_{\mathrm{max}}(P_{ABXY}(0101) + P_{ABXY}(1010) + P_{ABXY}(0011)),
\end{aligned}
\end{equation}

where $\mu_{\mathrm{min}}$ ($\mu_{\mathrm{max}}$) are lower (upper) bounds on the probability of the source to produce any given output $xy$. In our experiment, we assume that both sources are $\mu-SV$ sources and therefore we have $\mu_{\mathrm{min}}=(0.5-\mu)^2$, $\mu_{\mathrm{max}}=(0.5+\mu)^2$. In the special case of a perfect source (i.e. $\mu=0$), the above inequality corresponds to the well-known CHSH inequality (in Eberhard form \cite{Eberhard1993}). Note that as $\mu$ increases, the inequality \eqref{eq:mdl_inequality} becomes increasingly difficult to violate. Therefore, for any quality of the experiment there exists a maximal threshold bias $\mu$ of the sources that can be tolerated (compare Fig.~\ref{fig:S_vs_n_theory}).

More quantitatively, we are interested in the amount of randomness which can be extracted from the outputs $\mathbf{AB}$ of our device. It is well established that the number of extractable random bits is characterized by the conditional (smooth) min-entropy $H_{\mathrm{min}}^{\varepsilon}(\mathbf{AB}|\mathbf{XY}E)$ \cite{Renner2005, Tomamichel2013, Arnon-Friedman2016}, where $E$ denotes any side-information the adversary may hold about $\mathbf{AB}$. In our protocol we include the outputs of the source $\mathbf{XY}$ as part of the conditioning system. This means that our protocol remains secure even if the adversary learns the outputs of the SV source. This property is called privatization. The main task now is to find a bound on the smooth min-entropy. This is a challenging problem in general since the system $\mathbf{AB}$ is very large. To overcome this difficulty, we need a method to bound $H_{\mathrm{min}}^{\varepsilon}(\mathbf{AB}|\mathbf{XY}E)$ in terms of a simpler quantity which we can evaluate. This is exactly what is achieved by the entropy accumulation theorem \cite{Dupuis2020}. In rough terms, the entropy accumulation theorem states that \cite{Kessler2020}
\begin{equation}
    \label{eq:eat_bound}
    H_{\mathrm{min}}^{\varepsilon}(\mathbf{AB}|\mathbf{XY}E) \geq H(S_{\mu}) n - c\sqrt{n},
\end{equation}
where $H(S_{\mu})$ is the single-round von Neumann entropy as a function of the MDL violation $S_{\mu}$ and $c$ is a constant. Using the two-source randomness extractor, the condition in \eqref{eq:eat_bound} then guarantees that the final bit string $\mathbf{K}$ satisfies the given security definition \cite{Arnon-Friedman2016}.

{\it Requirements for the randomness extractor --}
For a randomness extractor to work in the context of device-independent randomness amplification, the two sources $\mathbf{Z}$ and $\mathbf{AB}$ of the extractor need to satisfy two conditions. The first condition is that both sources need to have sufficiently high min-entropy.
Since the bit string $\mathbf{Z}$ is produced by a SV source, we can directly bound its min-entropy by $H_{\mathrm{min}}(\mathbf{Z}) \geq -d \log\left(\frac{1}{2} + \mu\right)$. The min-entropy of $\mathbf{AB}$ on the other hand is lower-bounded by Eq.~\ref{eq:eat_bound}.
The second condition is that, from the adversary's perspective, the two inputs to the extractor (i.e., $\mathbf{AB}$ and $\mathbf{Z}$) need to be independent of each other. This condition essentially requires that the outputs $\mathbf{AB}$ remain separated from the QRNG so that there is no causal influence from $\mathbf{AB}$ to $\mathbf{Z}$. For a more formal treatment of this condition, see Supplementary Information section~V.
In summary, we produce two bit strings $\mathbf{AB}$ and $\mathbf{Z}$ which both have high min-entropy and are independent of each other from the adversary's point of view. These two properties together then allow for the extraction of a perfectly random bit string. A detailed description of the security proof can be found in Supplementary Information section~VI.

\ 

{\it Setup improvements --} 
In this experiment we aim to generate a Mbit-range output bit string, which requires to perform $\approx10^9$ Bell test trials. This number exceeds by 3 orders of magnitude the statistics presented in Ref.~\cite{Storz2023} demonstrating a loophole-free Bell inequality violation with superconducting circuits ($10^6$ trials).
This necessitates to maintain the calibrations and setup performance in a stable manner over the course of days, while continuously taking data. We have approached this challenge from both directions by shortening the required data acquisition time and engineering the control electronics to be stable over longer time scales. The former was achieved by increasing the repetition rate of the experiment from 12.5~kHz to 50~kHz, and by greatly decreasing the digital processing and network communication overhead time from $> 200 \mu$s/trial to $< 2 \mu$s/trial. These improvements have shortened the required data acquisition time to slightly less than nine hours. We have also majorly enhanced the stability of the room-temperature microwave control electronics and the steadiness of the pulse timing and triggering, which resulted in the performance stability illustrated by Fig.~\ref{fig:S_vs_time}. Technical details of the modifications and our timing and triggering schemes can be found in Supplementary Information section~II.

Alongside, in this experiment we achieve the CHSH S-value of $S=2.27$ (compare to $S=2.07$ in Ref.~\cite{Storz2023}). In addition to an increase stemming from the enhanced stability of control electronics and triggering, we have targeted the main factor limiting the fidelity of the Bell state, i.e. microwave photon losses in the quantum channel connecting the two sites. Specifically, we removed a cryogenic circulator, which was previously used to aid with thermalization of the intra-waveguide stray radiation and allowed measuring the temporal mode function of microwave photons emitted from either of the nodes. In addition to enabling bi-directional photon exchange across the channel, removing the circulator reduced the channel loss by about 3$\%$. As the 30-m superconducting aluminum waveguide connecting the nodes has negligible ($<1$ dB/km) losses, the majority of losses are due to the on-chip and PCB wiring, flexible cables connecting the sample mount to the waveguide, and adapter parts enabling coax-to-waveguide and coax-to-PCB coupling. We have therefore improved the PCB-to-waveguide wiring by decreasing the amount of adapter parts, and by replacing flexible coaxial cables by superconducting ones \cite{Storz2023a}. These changes have further cut the channel loss by about 3$\%$. In total, we estimate to have decreased the channel loss from about 19$\%$ to about 12-14$\%$.

\ 


\section {Supplemental information}

\section{I. Experimental devices and a two-node untrusted $\lambda$}

In this section, we describe the experimental equipment used in the device-independent randomness amplification implementation and specifically focus on what forms part of the untrusted device $\lambda$ introduced in the main text, and which devices we consider to be trusted. Brief schematics of individual components of the experimental setup and the division into an untrusted and trusted parts are illustrated in Figure~\ref{fig:device_lambda}.

First and foremost, $\lambda$ contains the cryogenic system which operates two remotely connected superconducting circuit quantum devices. As discussed in detail in Ref.~\cite{Schaer2023}, the system is composed of two dilution refrigerators at nodes A and B, housing the quantum devices, of a central cooling unit, and of a 30-meter-long cryogenic quantum link. The quantum link contains a rectangular Aluminium waveguide cooled to temperatures below 50~mK along the full distance, serving as a quantum channel between the two nodes.

The quantum devices operated at the nodes each contain a transmon-style qubit, coupled to control lines for performing gates along X, Y and Z axes of the corresponding Bloch sphere ($\sigma_x, \sigma_y$ and $\sigma_z$ operators), and a resonator-Purcell filter circuit used for reading out the state of the qubit using a microwave tone \cite{Walter2017}. The qubits are furthermore coupled to another resonator-Purcell filter circuit which connects them to the 30-meter-long quantum channel. This connection allows for the exchange of single microwave photons between the two nodes, and thereby for the generation of entanglement \cite{Kurpiers2017}. The qubits are operated at a frequency of about 7.8~GHz, and they show coherence times of about $T_1\sim$20~$\mu$s and $T_2^\star\sim$10~$\mu$s. We refer the interested reader to Ref.~\cite{Storz2023}, which uses the same quantum devices, for further information.

In addition, we consider all of the electronic setup for controlling and reading out the state of the qubits as part of the untrusted device $\lambda$ (with explicit exceptions mentioned below). This includes the arbitrary waveform generators and microwave generators used to prepare qubit control pulses and the readout signals, microwave equipment such as mixers and amplifiers, and the cables connecting the room-temperature electronics to the quantum devices in the cryostats. Hence the entirety of the quantum and cryogenic hardware, as well as the majority of the control electronics form part of the untrusted two-node black-box device $\lambda$.

The trusted part of the equipment, located outside of $\lambda$, is formed by the devices necessary to verify the device-independent protocol interface (see Box I in the main text) and ensure the closure of the locality loophole. First, the random number generators (RNGs) used to prepare the input bits are considered trusted, although public and imperfect, which reflects the assumption of them being Santha-Vazirani (SV) sources \cite{Santha1986}. Second, regarding the measurement process, we consider trusted the analog-to-digital converter (ADC) used to digitize the measurement signal, and the field-programmable gate array (FPGA) used to perform state discrimination and post-processing. Also outside of $\lambda$ is the equipment used to verify the space-time configuration of the experiment, relevant for closing the locality loophole. This includes the oscilloscopes used for the timing verification and to record reference pulses (see Fig.~\ref{fig:device_lambda}); further details about the real-time synchronization monitoring and the verification devices are given in Supplementary Information sections~II and~III. Finally, the computer storing the measurement results in a classical memory and performing the randomness extraction are part of the trusted equipment.

\begin{figure*}
  \begin{center}
  \includegraphics[width=1.8\columnwidth]{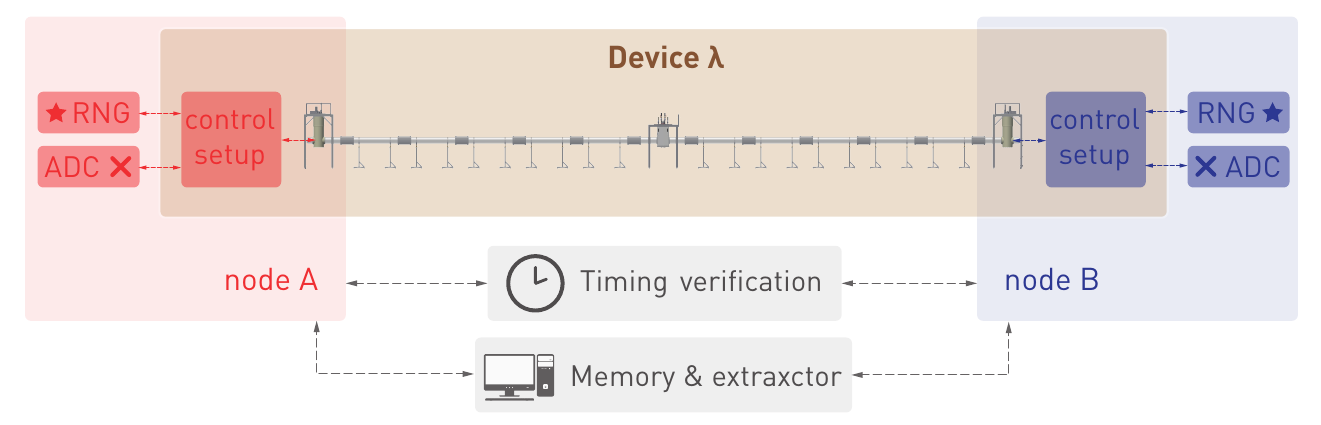}
  \caption{Division of the experimental setup into the untrusted (Device $\lambda$) and trusted (everything outside of $\lambda$) parts. Red and blue squares highlight the locations of nodes A and B of the two-node setup, which are separated by 30 metres. The untrusted part contains the entirety of the cryogenic enclosure and quantum devices, the room-temperature (RT) control setup and wiring. The only two trusted components of the RT electronics are a randomness source and an ADC, compare with Fig.~\ref{fig:protocol} in the main text. Timing verification equipment and classical memory also form part of the trusted equipment. Further details about timing verification are provided in Supplemental Information section~II.}
  \label{fig:device_lambda}
  \end{center}
\end{figure*}

\section{II. Triggering scheme, setup \& stability modifications}


The requirement to maintain sub-nanosecond synchronization precision over timescales of several days poses a significant practical challenge. Furthermore, the hardware specifications and datasheets of the microwave devices used in the setup only guarantee jitter on the order of the individual devices' granularity, which for a Tektronix AWG5014 comes to about 1.67~ns. Hence it has taken several iterations to achieve the reliable and stable synchronisation of the electronic devices used for the experiment, which shows effectively no jitter on week-long time windows. In the following, we discuss key aspects of our synchronisation scheme, including the distribution of a stable reference clock and trigger signals, the avoidance of jitter of electronic signals, and an automated timing verification scheme.


\paragraph{Reference clock and trigger --} We generate the main reference clock signal at the center of the cryogenic link using a two-stage process. First, a Rubidium clock with an output frequency of 10~MHz is used as a reference input to a dedicated microwave generator (MWG). The generator then outputs a stable 1~GHz signal, which we distribute to both end nodes A and B through dedicated coaxial cables of identical length, as illustrated in Figure \ref{fig:clock_trigger_scheme}. The symmetrical cabling scheme ensures an equal time delay and signal power at both nodes. Furthermore, the variations of electrical length of the cables due to global temperature variations in the laboratory are also identical. At the nodes, we daisy\nobreakdash-chain the microwave generators involved in the experiment using the 1~GHz reference input/output port of the used microwave sources, where the first MWG in the chain is synchronized to the reference clock signal in the middle. 

In the applications involving synchronizing a heterogeneous network of multiple digital and analog signal sources, one typically faces the problem of signal jitter, a property of digital electronics to have bi\nobreakdash-stable input trigger regions. We mitigate this problem by clocking the local electronic devices using frequency, phase and power-optimized reference signals generated by dedicated MWGs, which are locked to the main reference clock. This also helps minimizing local phase drifts, as further discussed below.

We generate the master trigger signal using a Tektronix AWG5014 located at the center of the link and distribute it to the two nodes, similarly to the reference clock. At each node, an AWG5014 receives the main trigger signal for triggering the relevant local electronic devices.

\paragraph{Timing verification --}
To monitor the relevant timings in the experiment and verify the closure of the locality loophole during the experiment execution, we operate oscilloscopes at each node. We use them to record the arrival times of reference pulses generated by AWGs and the RNG, output at a pre-defined moment in time, and sent through cables of a known constant length, see Figure~\ref{fig:clock_trigger_scheme}a. The inset in Figure~\ref{fig:clock_trigger_scheme}a displays the typical pulse sequence of a trial of the experiment, including microwave pulses sent to the qubit in the cryostat and reference pulses for the oscilloscope. The local oscilloscope records the reference pulses from the individual devices repeatedly over the course of several seconds, and stores the envelope of the captured traces, i.e. the maximum and minimum voltage measured at each point in time throughout the current acquisition period (Figure~\ref{fig:clock_trigger_scheme}c). In the timing measurements presented in Supplementary Information sections~III and~IV, following a conservative approach, we use the envelopes which maximize the total protocol duration, therefore calculating an upper bound on the actual duration of each individual trial.

\paragraph{Phase drift measurements --}
Using the setup described above, we determined the phase drifts of the individual devices under the optimized clocking and triggering scheme, characterizing the duration of the absence of jitter in the system. The absolute drifts summarized in Table \ref{tab:drift_and_granularity} are measured relative to the AWG5014 at the nodes A \& B, we took the larger of the two measured values. The AWGs and FPGAs have small drift compared to their granularity. The largest drift belongs to the QRNGs, limiting the system to be in a jitter\nobreakdash-free configuration for about 2.5 days. To make full use of a stable configuration for the main experiment presented in this publication, we adjust the phases of the reference signals clocking the AWGs in a way so that the trigger arrival time on all devices is in the middle of their two granularity points before starting acquiring the dataset.

\begin{table}[h]
    \centering
    \begin{tabular}{|l|l|l|} 
        \hline
        \multirow{2}{*}[4pt]{Device} & \multirow{2}{*}[4pt]{Granularity} & \multirow{2}{*}[4pt]{Relative drift} \\
        & in ns & in ps/hr \\
        \hline
        AWG5014 (other node)& 1.66 & 0.1  \\
        FPGA & 8 & 0.1  \\
        AWG70K & 6.4 & 0.2  \\
        QRNG  & 2 & 16.9  \\
        \hline
    \end{tabular}
    \caption{Granularity and drift properties of the electronic devices used to control or readout the qubit. The drift measurements are taken about a course of 50 hours and are a linear fit to the obtained data from the oscilloscope. 
    }
    \label{tab:drift_and_granularity}
\end{table}

\begin{figure*}
  \begin{center} 
   \includegraphics[width=2\columnwidth]{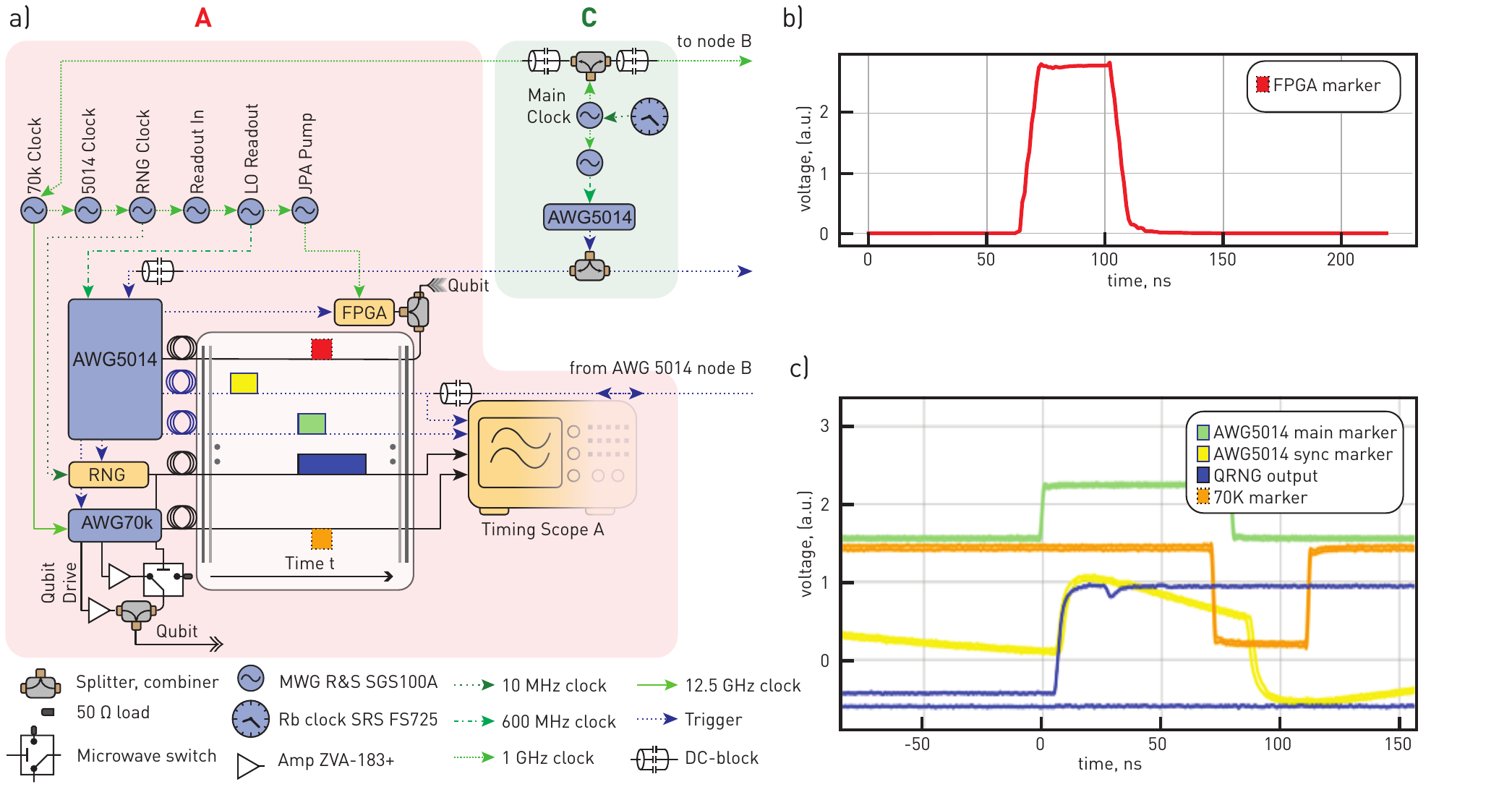}
  \caption{a) Clock and trigger distribution scheme for the synchronization of the room-temperature electronics at node A and center C, the scheme is identical at node B. Inset: pulses and connections to oscilloscope and FPGA to monitor and measure the output delays of various devices. Yellow: auxiliary pulse from the master device (AWG 5014) at another node; Green: timing reference pulse. Blue: QRNG output signal; Orange: a negative image of a fixed-time synchronization pulse from the AWG 70k. b,c) An example of the measured waveforms on an FPGA and an oscilloscope. The signal on the FPGA corresponds to the pulse output from AWG5014 at a fixed time and is used to synchronize the FPGA to the global lab timeframe. The oscilloscope pulses correspond to the ones shown in the inset a) with same color code. The green reference pulse defines the zero time; the yellow pulse travelling from another node is used to ascerain the nodes' synchronization. The blue and orange pulses provide the information about the RNG and AWG70k reference frame synchronization. Waveforms with dotted lines in a) are only generated and measured between the acquisition blocks, while the others are always available}
  \label{fig:clock_trigger_scheme}
  \end{center}
\end{figure*}

\section{III. Timing Stability}
\label{app:time_stab}

As the main experimental sequence is sufficiently extended in time (spanning more than 8 hours), in addition to full timing calibration before and timing verification after the experiment, we employ interleaved timing checks ensuring that the locality loophole is closed during the entire duration of the experiment. To achieve that, we install two oscilloscopes, one at each node of the setup, and connect them to the auxiliary outputs of the three room-temperature control devices at each node: the master device of the node (Tektronix AWG 5014), the device we use to generate the qubit control pulses (Tektronix AWG 70k), and the QRNG. Finally, we connect the fourth port of the oscilloscope via a splitter to a cable connecting the two nodes A and B in order to monitor the synchronisation between the nodes.

\begin{figure}
  \begin{center} \includegraphics[width=1\columnwidth]{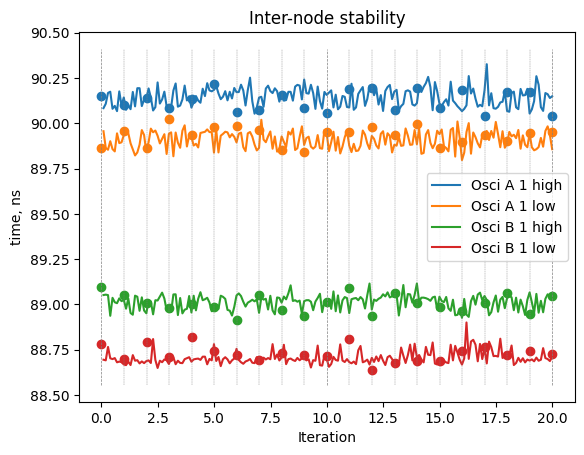}
  \caption{\textbf{Inter-node timing stability and synchronisation} Data for measurements of an auxiliary pulse arriving from another node to an oscilloscope installed at node A (blue and orange) and B (green and red). Each datapoint corresponds to several thousand triggers recorded in the envelope mode on an oscilloscope; Plotted are both high (late) and low (early) points of the envelope trace crossing the pre-defined threshold. Presence of jitter would manifest as the high and low points corresponding to the same device being different by the device granularity, 1.66~ns. The existing difference between the high and low envelope traces reflects the limited sampling rate of the oscilloscope (2.5 GS/s) and possible jitter on the order of 100~ps, which is within our uncertainly calculations. Full circles correspond to datapoints taken between the data blocks, and continuous lines correspond to the on-the-fly timing verification measurements.}
  \label{fig:intra_node}
  \end{center}
\end{figure}

Before the first and after every block of data during the data acquisition (recall, each block contains $4^{13}$ ($\sim 67$~million) trials) we perform the following timing verification routine: we program a reference pulse sequence (inset in Fig.~\ref{fig:clock_trigger_scheme}a) to every device in the system, and acquire several thousand triggers over fifteen seconds on both oscilloscopes. We process the data by taking the envelope of received signals (a result of such data acquisition is shown in Fig.~\ref{fig:clock_trigger_scheme}c), and for a certain threshold 
we plot the low and high envelope point. As we trigger a local oscilloscope on the signal arriving from the master AWG of the corresponding node, observing a signal arriving from another device allows monitoring relative time stability of the two devices. In addition, presence (or absence) of a time difference between the low and high envelopes of the same signal signifies presence (or absence) of trigger jitter -- another effect which could prevent us from closing the locality loophole. Note, that time differences sufficiently below 1~ns are due to the pulse width and limited timing resolution of the measurement device and are not a sign of jitter: trigger or device jitter would lead to the time difference of at least the smallest granularity of the control devices, which is 1.66~ns in our case (Tektronix AWG 5014).

\begin{figure*}
  \begin{center} \includegraphics[width=1\columnwidth]{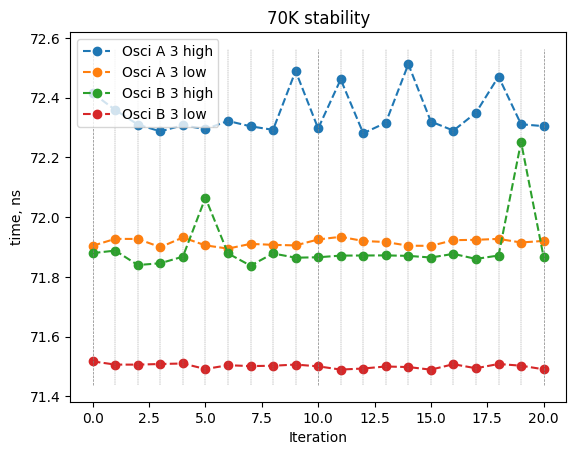}\includegraphics[width=1\columnwidth]{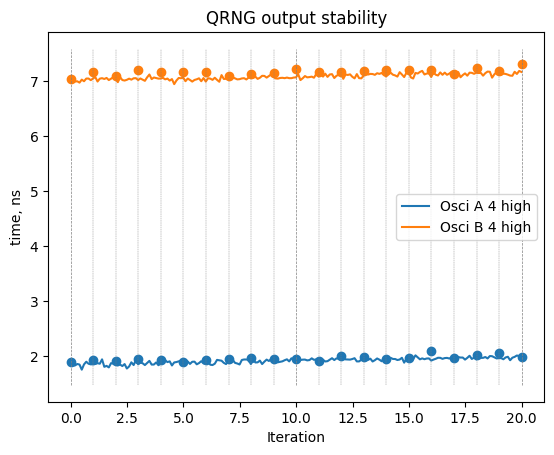}
  \caption{\textbf{Timing at a single node} (left) Data showing the arrival time of a timing calibration pulse from the AWG70k used to control qubit rotations at node A (blue and orange) and B (green and red). As in Fig.~\ref{fig:intra_node}, each datapoint corresponds to several thousand triggers recorded in the envelope mode on an oscilloscope; Plotted are both high (late) and low (early) points of the envelope trace crossing the pre-defined threshold. These timings are only measured between the data acquisition blocks. (right) Data showing the arrival time of a QRNG switching pulse and hence defining the starting time of each trial. Compare to Table~\ref{tab:timing}. Full circles correspond to datapoints taken between the data blocks, and continuous lines correspond to the on-the-fly timing verification measurements.}
  \label{fig:Node_stability}
  \end{center}
\end{figure*}

Figure~\ref{fig:intra_node} shows the inter-node timing and its stability. The key information one obtains from the figure is the deviations between the maximum and minimum for each curve (see figure caption for details) are bounded by at most 200ps and the curves are stable in time. In a similar fashion, Fig.~\ref{fig:Node_stability}a demonstrates the relative time stability of the two AWGs on each node, and figure~\ref{fig:Node_stability}b shows the QRNG output timing stability. Note that as QRNG does not output a pulse, but rather switches the state and stays in the new state, we could not measure the low envelope of the switching channel (and hence observe possible jitter in the positive time direction). An example of such measurement is also shown in Fig.~\ref{fig:clock_trigger_scheme}c. In other words, if during one of the trials the QRNG switches later than it should, it would not appear on our measurement. That does not lead to opening the locality loophole however, as it only makes the length of the algorithm during that trial shorter (and decreases the state preparation quality, of course).

Evidently, all three figures show neither jitter nor instability in the timings during the experimental execution. From Fig.~\ref{fig:intra_node} one can remark that our two nodes are not perfectly in sync, but rather their time frames differ by 0.55~ns. This level of clocking mismatch is reasonable as our sync precision is limited by approximately half the granularity of the master AWGs at each node. We take this mismatch into account in the overall time budget of the experiment by conservatively adding it to both $A\rightarrow B$ and $B\rightarrow A$ protocol durations. Hence from the steadiness of the timing and clock matching precision we can conclude that the locality loophole was closed during the whole data acquisition run.

\section{IV. Timing Verification}
\label{app:time_veri}

In this section, we present our method of characterising the space-time configuration of the experiment, in order to ensure that the locality loophole is closed.
As mentioned in the main text, we define the start time $t_\star$ of each trial of the randomness amplification protocol as the time at which the local random number generators start generating a random input bit (see Fig.~\ref{fig:protocol_experiment}). The stop time $t_\times$ marks the time when the untrusted qubit readout signal arrives at the input of the analog-to-digital converter (ADC), part of the trusted laboratory.

The locality loophole is closed when both of the following conditions hold: the basis choice at node A must not be able to influence the measurement at node B using influences propagating at most at the speed of light $c$, and vice versa. In order to ensure that, we must precisely know where and when the start and stop events happen.

In a first step, we measure the spatial distance between the locations corresponding to the start event at one node and the stop event at the other node, in a scheme discussed in Ref.~\cite{Storz2023}. We find the shorter of these two distances to be the one from the start event at node B to the stop event at node A, measuring $d=32.928\,\rm{m}$.
This defines the Bell distance, and thereby the available duration during which the exchange of classical information exchange between the two nodes is forbidden, $t_d=d/c=109.83$~ns. 

In a second step, we verify that the actual duration of the protocol $t_\mathrm{protocol}=t_\times-t_\star$ is shorter than the budget $t_d$ imposed by the Bell distance.
We perform the analysis for both critical durations, i.e. for $t_{\mathrm{protocol}, A\rightarrow B}$ referring to the duration between the start event at A and the stop event at B, and vice versa. In the following, for simplicity, we focus on $t_{\mathrm{protocol}, A\rightarrow B}$, and the calculation of the $t_{\mathrm{protocol}, B\rightarrow A}$ follows accordingly.

At the beginning of the experiment, it is critical to establish a common reference frame for the devices involved in the experiment. For this purpose, we install an oscilloscope at each node (part of the trusted setup), to detect the arrival time of specific reference pulses output by electronic devices in the setup (part of $\lambda$). We define the reference time $t_{\mathrm{0,Osci A}}$ as the moment when oscilloscope A detects a reference pulse from the local master arbitrary waveform generator (AWG), sent through a cable with a well-characterized electrical length. We refer the interested reader to Ref.~\cite{Storz2023} for more details.

In a next step, we determine the protocol duration $t_{\mathrm{protocol}, A\rightarrow B}$.
As we do not have direct experimental access to the start and stop events defined above, we extract $t_{\mathrm{protocol}, A\rightarrow B}$ by performing a set of measurements of the duration of the individual segments of the protocol. For this purpose, it is instructive to express the protocol duration in terms of its individual constituents,

\begin{align}
    t_{\mathrm{protocol}, A\rightarrow B} &= t_{\times B} - t_{\star, A} \\
     &= t_{\mathrm{detection, B}} + t_{\Delta B}  \\
     &+ t_{\Delta AB} + t_{\mathrm{latch, A}} + t_{\mathrm{RNG, A}}.
    \label{eq:6_timing}
\end{align}

Figure~\ref{fig:timing} provides a graphical illustration of the individual segments of the protocol.
Here $t_{\mathrm{detection, B}}$ marks the length of the data acquisition, or the difference of the time at which the ADC at node B starts acquiring data and the time when it stops ($t_{\times B}$). We can experimentally control this duration by programming the field-programmable gate array (FPGA), which is connected to the ADC, to have zero acquisition weights after the specified time. Typically, however, the FPGA is not perfectly synchronized to the node reference frame $t_{\mathrm{0,Osci B}}$. To account for differences between the internal reference frame of the local FPGA, $t_{\mathrm{0,FPGA B}}$, and the node reference frame $t_{\mathrm{0,Osci B}}$, we include their offset $t_{\Delta B}=t_{\mathrm{0,FPGA B}}-t_{\mathrm{0,Osci B}}$ in the expression above. This offset is measured in a separate calibration experiment, where we compare the arrival time of a microwave square pulse recorded by the local oscilloscope and by the local FPGA.

Next, $t_{\mathrm{RNG, A}}$ denotes the time it takes for the local RNG to generate a random bit. The random number generators used in this experiment \cite{Abellan2015} produce random bits based on the principle of phase diffusion of a laser, founded on the underlying quantum mechanical process of spontaneous emission. More specifically, $t_{\mathrm{RNG, A}}$ is the duration between the laser injection current crossing the lasing threshold and the moment the RNG outputs the random bit in the form of a differential voltage, as described in Ref.~\cite{Abellan2015}. We have more direct experimental access to the latching event, and define the time between the RNG latching a random bit and the node reference time $t_{\mathrm{0,Osci A}}$ as $t_{\mathrm{latch, A}}$ (Fig.~\ref{fig:timing}).

Finally, we take into account the imperfect synchronization of the two node reference frames, $t_{\Delta AB}$. We determine this synchronisation offset by sending a reference pulse from an AWG at node A to the oscilloscope at node B, and vice versa, through the same cables. This allows us to compare the respective arrival times of the pulses, and thereby determine $t_{\Delta AB}$.

We measure all quantities during the randomness amplification experiment via a dedicated timing test sequence. Furthermore, we verify the absence of jitter of electronic devices involved in the setup continuously during the experiment, by monitoring the arrival time of reference pulses output by the relevant devices on the local oscilloscopes.

\begin{table}
		\centering
		\begin{tabular}{|c|c|c|c|}
			\hline
			Quantity at A & Value (ns) & Quantity at B & Value (ns) \\
			\hline
			$t_{\mathrm{RNG, A}}$ & 17.1 & $t_{\mathrm{RNG, B}}$ & 17.1\\
			$t_{\mathrm{latch, A}}$ & 2.0 & $t_{\mathrm{latch, B}}$ & 7.2\\
			$t_{\Delta A}$ & 5.0 & $t_{\Delta B}$ & 2.0\\
			$t_{\Delta AB}$ & 0.55 &  & \\
			$t_{\mathrm{detection, A}}$ & 86.0 & $t_{\mathrm{detection, B}}$ & 94.0\\
			\hline
			$t_{\mathrm{protocol}, A\rightarrow B}$ & 106.65 & $t_{\mathrm{protocol}, B\rightarrow A}$ & 106.45\\
			\hline
		\end{tabular}
		\caption{Measured durations of individual segments of the protocol, as described in the text. }
		\label{tab:timing}
	\end{table}

As no jitter was detected during the main experiment, and as the corresponding timing measurements (section~\hyperref[app:time_stab]{III}) yielded consistent results (up to the measurement precision), we ascertain the locality loophole to be closed. From the measured individual protocol segment durations introduced above, see Table~\ref{tab:timing}, we conclude that the total duration of the protocol was no more than $t_{\mathrm{protocol}, A\rightarrow B}=106.7$~ns and therefore $t_d-t_{\mathrm{protocol}, A\rightarrow B}=3.1$~ns shorter than the budget imposed by the Bell distance. We note that this margin is about 50$\%$ higher than for the loophole-free Bell test presented on the same setup \cite{Storz2023}, and that it is, constituting about 2.8$\%$ of the available time budget, in a similar regime as other loophole-free Bell tests \cite{Hensen2015,Giustina2015a,Shalm2015,Li2018i}.

\begin{figure*}
  \begin{center} \includegraphics[width=2.8\columnwidth, angle =-90]{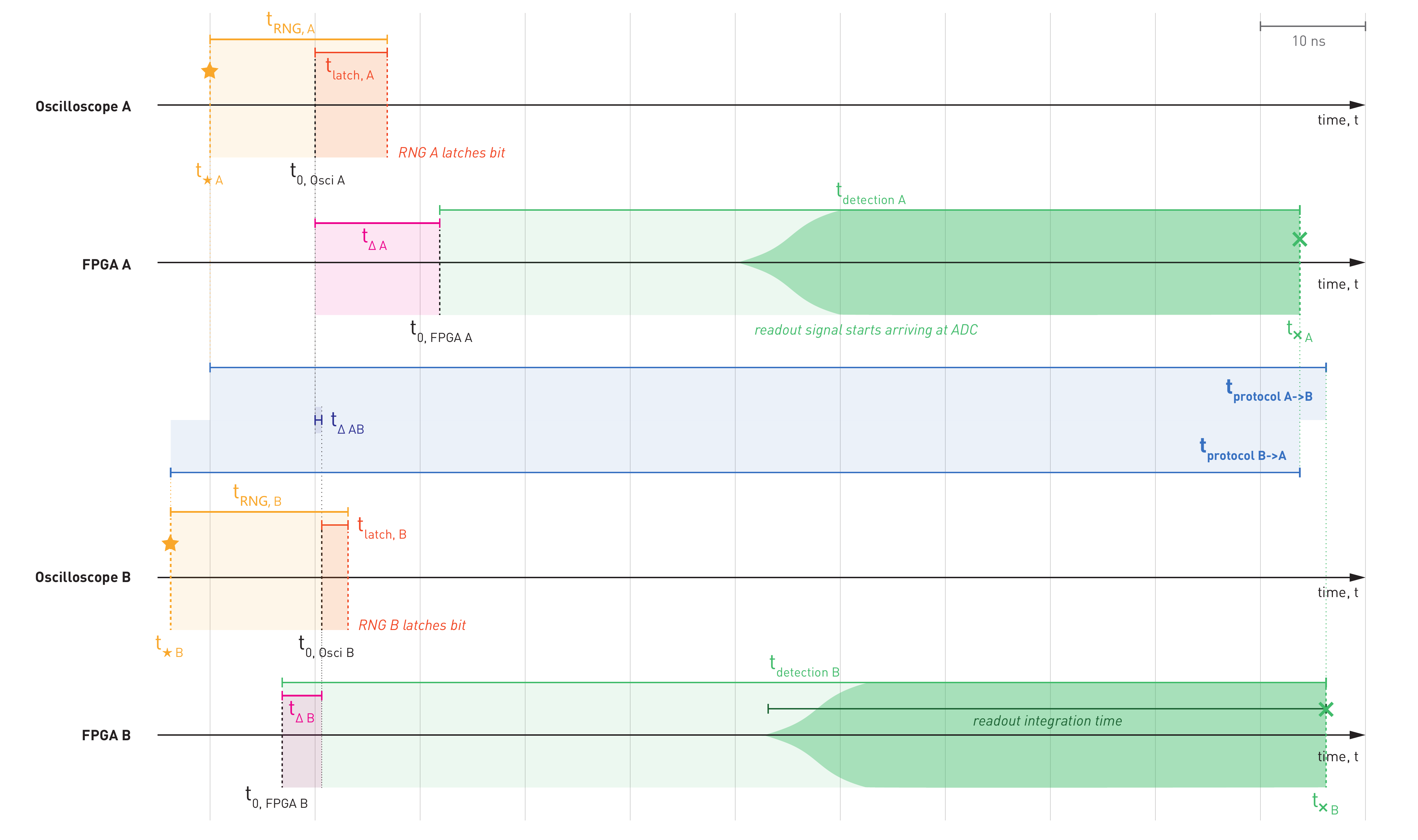}
  \caption{Timing overview of the randomness amplification experiment. See section~\hyperref[app:time_veri]{IV} for details.}
  \label{fig:timing}
  \end{center}
\end{figure*}

\section{V. Assumptions} 
\label{app:assumptions}

In this section we explicitly list all of the assumptions required for the security proof of our device-independent randomness amplification protocol.
Our source produces a sequence of $n$ random bits $X_{1}, \ldots, X_{n}$. We will denote the bit string consisting of the first $i$ bits as $X^{i} = X_{1} \ldots X_{i}$. We use similar notation for other bit strings such as $A^{i}, B^{i}, Y^{i}$ or $Z^{i}$. We denote by $I(A:B|C) = H(A|C) - H(A|BC)$ the conditional mutual information which can be loosely thought of as quantifying the amount of correlations between $A$ and $B$ if one already knows $C$.

We make the following assumptions:

\begin{enumerate}
    \item Quantum mechanics is correct and complete. In particular, the adversary can be described using quantum mechanics.
    \item The device $\lambda$ is shielded from the adversary, i.e., it does not leak any unwanted information. In particular, the outputs $A^{n}$ and $B^{n}$ are not leaked to the adversary.
    \item The classical information processing devices are trusted. In particular, this includes the ADC, the timing verification and the computer used for classical post-processing.
    \item The two sources are, possibly correlated, SV sources. This includes the assumption that Eve only has classical side-information about the sources (her side-information about the device $\lambda$ itself is allowed to be quantum).
    \item The outputs $X_{i}Y_{i}$ of the two sources in round $i$ are not influenced by the outputs $A^{i-1}B^{i-1}$ of the device from previous rounds. This assumption is justified since the QRNG is separated from the outputs of the devices \cite{Kessler2020}. More formally, we require that
    \begin{equation}
        I(X_{i}Y_{i}:A^{i-1}B^{i-1}|X^{i-1}Y^{i-1}E) = 0.
    \end{equation}
    Note that this condition does not require that $X_{i}Y_{i}$ is independent from $A^{i-1}B^{i-1}$ as they may be correlated through $X^{i-1}Y^{i-1}E$. In a similar fashion we require that
    \begin{equation}
        I(Z^{d}:A^{n}B^{n}|X^{n}Y^{n}E) = 0,
    \end{equation}
    where $Z^{d}$ is the second input to the randomness extractor which is drawn from the same SV source as $X^n$.
\end{enumerate}

\section{VI. Security proof}
Here we present the security proof for the implemented device-independent randomness amplification protocol. The proof holds under the assumptions given in Supplementary Information section V. The proof is based on combining the results from \cite{Kessler2020} with those from \cite{Arnon-Friedman2016,Dodis2004}. The main ingredient is the following bound on the smooth min-entropy derived in \cite{Kessler2020} which we adapt to our notation:

\begin{theorem} \label{thm:dira_hmin} (Theorem 4 in \cite{Kessler2020})
    Let $\rho$ be the state generated by the protocol after step 1, $S_{\mu} \in \mathbb{R}$, $\Omega$ be the event that $S_{\mu,\mathrm{obs}} \geq S_{\mu}$, $\varepsilon_{\mathrm{EA}}, \varepsilon_{s} \in (0, 1)$ be arbitrary, and $\rho_{|\Omega}$ be the state conditioned on $\Omega$. Then, under the assumptions outlined in section~\hyperref[app:assumptions]{V}, either $\mathrm{Pr}[\Omega] \leq \varepsilon_{\mathrm{EA}}$ or
    \begin{equation}
        H_{\mathrm{min}}^{\varepsilon_{s}}(\mathbf{AB}|\mathbf{XY}E)_{\rho_{|\Omega}} \geq n \eta(\varepsilon_{s}, \varepsilon_{\mathrm{EA}}, S_{\mu}, n, \mu),
    \end{equation}
    with the function $\eta$ as defined in \cite{Kessler2020}.
\end{theorem}

Next, we will introduce the formal extractor definitions.

\begin{definition}
    Let $m, n_{1}, n_{2} \in \mathbb{N}_{>0}$, $k_{1}, k_{2} \in \mathbb{R}$, and $\varepsilon \in (0, 1)$. A function $\mathrm{Ext}: \{0, 1\}^{n_{1}} \times \{0, 1\}^{n_{2}} \rightarrow \{0, 1\}^{m}$ is called a $(k_{1}, k_{2}, \varepsilon)$ $2$-source extractor if for all independent random variables $X_{1}$ and $X_{2}$ with min-entropies $H_{\mathrm{min}}(X_{1}) \geq k_{1}$ and $H_{\mathrm{min}}(X_{2}) \geq k_{2}$, the following condition holds:
    \begin{equation}
        \frac{1}{2} \|\mathrm{Ext}(X_{1}, X_{2}) - U_{m}\|_{1} \leq \varepsilon,
    \end{equation}
    where $U_{m}$ is the uniform distribution on $\{0, 1\}^{m}$. The extractor is called $X_{1}$-strong if in addition
    \begin{equation}
        \frac{1}{2} \|\mathrm{Ext}(X_{1}, X_{2})X_{1} - U_{m}X_{1}\|_{1} \leq \varepsilon.
    \end{equation}
\end{definition}

Many constructions for extractors exist. Given the large size of our dataset ($n \sim 10^{9}$), we need an extractor which can be efficiently implemented. We choose the following construction based on non-cyclic right shift matrices from \cite{Dodis2004}.

\begin{theorem} \label{thm:dodis_extractor} (Lemma 4 in \cite{Dodis2004})
    Let $m, n \in \mathbb{N}$ and $k_{1}, k_{2} \in \mathbb{R}$. Consider the function
    \begin{equation}
    \begin{aligned}
        \mathrm{Ext}: \{0, 1\}^{n} \times \{0, 1\}^{n} &\rightarrow \{0, 1\}^{m} \\
            (\mathbf{x}, \mathbf{y}) & \mapsto (\mathbf{x}^{T} A_{1} \mathbf{y}, \ldots, \mathbf{x}^{T} A_{m} \mathbf{y}),
    \end{aligned}
    \end{equation}
    where $A_{i}$ is the matrix representing a non-cyclic right shift by $i-1$ positions, $\mathbf{x}^{T}$ represents the transpose of $\mathbf{x}$, and all additions are modulo two. Then, the function $\mathrm{Ext}$ is a $(k_{1}, k_{2}, \varepsilon)$ two-source extractor with $\varepsilon = 2^{-\frac{1}{2}(k_{1} + k_{2} - n - 2m)}$. Furthermore, $\mathrm{Ext}$ is $X_{1}$-strong.
\end{theorem}

Next, we state the extractor definition used for our device-independent randomness amplification protocol.

\begin{definition} (Definition 8 in \cite{Arnon-Friedman2016})
    Let $m, n_{1}, n_{2} \in \mathbb{N}_{>0}$, $k_{1}, k_{2} \in \mathbb{R}$, and $\varepsilon \in (0, 1)$. A function $\mathrm{Ext}: \{0, 1\}^{n_{1}} \times \{0, 1\}^{n_{2}} \rightarrow \{0, 1\}^{m}$ is called a $(k_{1}, k_{2}, \varepsilon)$ quantum-proof $2$-source extractor in the Markov model if for all tripartite classical-classical-quantum states $\rho_{X_{1} X_{2} C}$ satisfying the Markov chain condition $I(X_{1}:X_{2}|C) = 0$ with min-entropies $H_{\mathrm{min}}(X_{1}|C)_{\rho} \geq k_{1}$ and $H_{\mathrm{min}}(X_{2}|C)_{\rho} \geq k_{2}$, the following condition holds:
    \begin{equation}
        \frac{1}{2} \|\rho_{\mathrm{Ext}(X_{1}, X_{2})C} - \omega_{m} \otimes \rho_{C}\|_{1} \leq \varepsilon,
    \end{equation}
    where $\omega_{m}$ is the maximally mixed state on $m$ qubits. The extractor is called $X_{1}$-strong if in addition
    \begin{equation}
        \frac{1}{2} \|\rho_{\mathrm{Ext}(X_{1}, X_{2})X_{1} C} - \omega_{m} \otimes \rho_{X_{1} C}\|_{1} \leq \varepsilon,
    \end{equation}
\end{definition}

The second core ingredient is the following result from \cite{Arnon-Friedman2016} which shows that any two-source extractor remains secure against quantum side-information with adjusted parameters.

\begin{theorem} \label{thm:extr_reduction} (Theorem 2 in \cite{Arnon-Friedman2016} (for $l=2$))
    Any $(k_{1}, k_{2}, \varepsilon)$ $X_{1}$-strong $2$-source extractor with output length $m$ is a $\left( k_{1} + \log \frac{1}{\varepsilon}, k_{2} + \log \frac{1}{\varepsilon}, \sqrt{3 \varepsilon 2^{m-2}} \right)$ $X_{1}$-strong quantum-proof$\;2$-source extractor in the Markov model.
\end{theorem}

Since Theorem \ref{thm:dira_hmin} only provides a bound on the smooth min-entropy, we need to show that the extractor $\mathrm{Ext}$ from Theorem \ref{thm:dodis_extractor} is still a good extractor when we only have a bound on $H_{\mathrm{min}}^\varepsilon$ instead of $H_{\mathrm{min}}$. This is the content of the next theorem.

\begin{theorem} \label{thm:smooth_extr} (Lemma 17 in \cite{Arnon-Friedman2016})
    Let $\varepsilon_{1}, \varepsilon_{2}, \delta_{1}, \delta_{2} \in (0, 1)$ be arbitrary and $\mathrm{Ext}: \{0,1\}^{n_{1}} \times \{0,1\}^{n_{2}} \rightarrow \{0,1\}^{m}$ be a $\left( k_{1} - \log \frac{1}{\varepsilon_{1}}, k_{2} - \log \frac{1}{\varepsilon_{2}}, \varepsilon \right)$ quantum-proof two-source extractor in the Markov model. Then for any classical-classical-quantum state $\rho_{X_{1}X_{2}C}$ satisfying the Markov chain condition $I(X_{1}:X_{2}|C) = 0$ with $H_{\mathrm{min}}^{\delta_{1}}(X_{1}|C)_{\rho} \geq k_{1}$ and $H_{\mathrm{min}}^{\delta_{2}}(X_{2}|C)_{\rho} \geq k_{2}$,
    \begin{equation}
        \frac{1}{2} \|\rho_{\mathrm{Ext}(X_{1}, X_{2}) C} - \omega_{m} \otimes \rho_{C}\|_{1} \leq 6 \delta_{1} + 6 \delta_{2} + 2 \varepsilon_{1} + 2 \varepsilon_{2}  + 2 \varepsilon.
    \end{equation}
\end{theorem}

\begin{remark}
    As noted in \cite{Arnon-Friedman2016} it can be readily checked that an analogous statement also holds for $X_{1}$-strong extractors.
\end{remark}

\begin{theorem} \label{thm:security_thm}
    Let $\rho$ be the state generated at the end of the protocol, $\varepsilon, \varepsilon_{\mathrm{EA}}, \varepsilon_{s} \in (0, 1)$ be arbitrary, $S_{\mu} \in \mathbb{R}$, and the output size $m \in \mathbb{N}$ be such that
    \begin{equation}
    \begin{aligned} \label{eq:extractor_condition}
        m \leq & \frac{n}{6}\left(\eta(\varepsilon_{s}, \varepsilon_{\mathrm{EA}}, S_{\mu}, n, \mu) + 2\log \frac{1}{1/2 + \mu} - 2 \right) \\
        &- \log \frac{6}{\varepsilon - 6\varepsilon_{s}}.
    \end{aligned}
    \end{equation}
    Let $\Omega$ be the event that $S_{\mu, \mathrm{obs}} \geq S_{\mu}$. Then, under the assumptions outlined in section~\hyperref[app:assumptions]{V}, either $\mathrm{Pr}[\Omega] \leq \varepsilon_{\mathrm{EA}}$ or
    \begin{equation} \label{eq:tr_distance_bound}
        \frac{1}{2} \| \rho_{\mathbf{KXYZ}E|\Omega} - \omega_{m} \otimes \rho_{\mathbf{XYZ}E|\Omega} \|_{1} \leq \varepsilon.
    \end{equation}
\end{theorem}

\begin{remark}
    Note that the condition in (\ref{eq:tr_distance_bound}) says that the output of the extractor looks random even if the adversary learns the output of the source ($\mathbf{XYZ}$). This property is also called \emph{privatization}.
\end{remark}

\begin{proof}
    Let $\rho$ be such that the probability of aborting is less than $1 - \varepsilon_{\mathrm{EA}}$ (otherwise the statement of the Theorem follows trivially). By applying Theorems \ref{thm:dodis_extractor}, \ref{thm:extr_reduction} and \ref{thm:smooth_extr}, a direct calculation shows that the condition in equation (\ref{eq:extractor_condition}) implies that our extractor is secure with security parameter $\varepsilon$ as long as
    \begin{equation}
    \begin{aligned}    
        H_{\mathrm{min}}(\mathbf{Z}|\mathbf{XY}E)_{\rho_{\Omega}} \geq& 2n \log \frac{1}{1/2 + \mu} \\
        H_{\mathrm{min}}^{\varepsilon_{s}}(\mathbf{AB}|\mathbf{XY}E)_{\rho_{|\Omega}} \geq& n \eta(\varepsilon_{s}, \varepsilon_{\mathrm{EA}}, S_{\mu}, n, \mu).
    \end{aligned}
    \end{equation}
    The first inequality is satisfied by the condition that our source is an SV-source with bias $\mu$ (see also Lemma 7 in \cite{Kessler2020}). The second inequality is guaranteed by Theorem \ref{thm:dira_hmin}.
\end{proof}

The previous section assumed that the key size is fixed before the protocol begins. In order to maximize the amount of randomness that can be extracted, one would like to adjust the output size depending on the violation of the MDL-inequality (compare Box I). This can be done using the following Theorem
\begin{theorem} Let $\varepsilon, \varepsilon_{s} \in (0, 1)$ be arbitrary. Let $S_{\mu,\mathrm{max}} \in \mathbb{R}$, $M \in \mathbb{N}$ and set $\Delta S_{\mu} = S_{\mu,\mathrm{max}}/M$. Choose the output size as
    \begin{equation}
    \begin{aligned}
        m'(S_{\mu,\mathrm{obs}}) = &\frac{n}{6}\left( \eta \left( \varepsilon_{s}, \varepsilon, \min \left\{ S_{\mu, \mathrm{max}}, \floor*{\frac{S_{\mu,\mathrm{obs}}}{\Delta S_{\mu}}} \Delta S_{\mu} \right\}, n, \mu \right) \right. \\
        & \left. + 2 \log \frac{1}{1/2 + \mu} - 2 \right) - \log \frac{6}{\varepsilon - 6 \varepsilon_{s}} \\
        m(S_{\mu,\mathrm{obs}}) =& \max \{ \floor{m'(S_{\mu,\mathrm{obs}})}, 0 \}.
    \end{aligned}
    \end{equation}
    Then, under the assumptions in section~\hyperref[app:assumptions]{V}, the protocol is $M \varepsilon$ secure.
\end{theorem}
\begin{proof}
    For $i \in \{1, \ldots, M - 1\}$, define $\Omega_{i}$ as the event when the observed MDL violation satisfies $i \Delta S_{\mu} \leq S_{\mu, \mathrm{obs}} < (i + 1) \Delta S_{\mu}$, $\Omega_{0}$ as the event when $S_{\mu,\mathrm{obs}} < \Delta S_{\mu}$, and $\Omega_{M}$ as the event when $S_{\mu, \mathrm{obs}} \geq S_{\mu, \mathrm{max}}$. Define $m_{i} = m(i\Delta S_{\mu})$, then the security parameter can be bounded by
    \begin{equation}
    \begin{aligned}
        & \sum_{i=0}^{M} \mathrm{Pr}[\Omega_{i}] \frac{1}{2} \| \rho_{\mathbf{KXYZ}E|\Omega_{i}} - \omega_{m_{i}} \otimes \rho_{\mathbf{XYZ}E|\Omega_{i}} \|_{1} \leq M \varepsilon,
    \end{aligned}
    \end{equation}
    where we used that for $i = 0$ the trace distance vanishes (because $m_{0} = 0$ and hence $\mathbf{K}$ is trivial) and used Theorem \ref{thm:security_thm} to bound the terms for $i = 1, \ldots, M$ by $\varepsilon$.
\end{proof}

\section{VII. Extracted random bits}
 Figure~\ref{fig:bits} shows the main output of this experiment: the 20,431,465 $\epsilon$-random bits produced via device-independent randomness amplification. Each pixel represents a single extracted bit.
Furthermore, as a visual demonstration of the random nature of the bits, in Figure~\ref{fig:sheep} we present the image of a Valais Blacknose sheep (left), and the same image after one-time padding it with a set of the generated bits from this experiment (right).

\begin{figure*}
  \begin{center} \includegraphics[width=2.0\columnwidth]{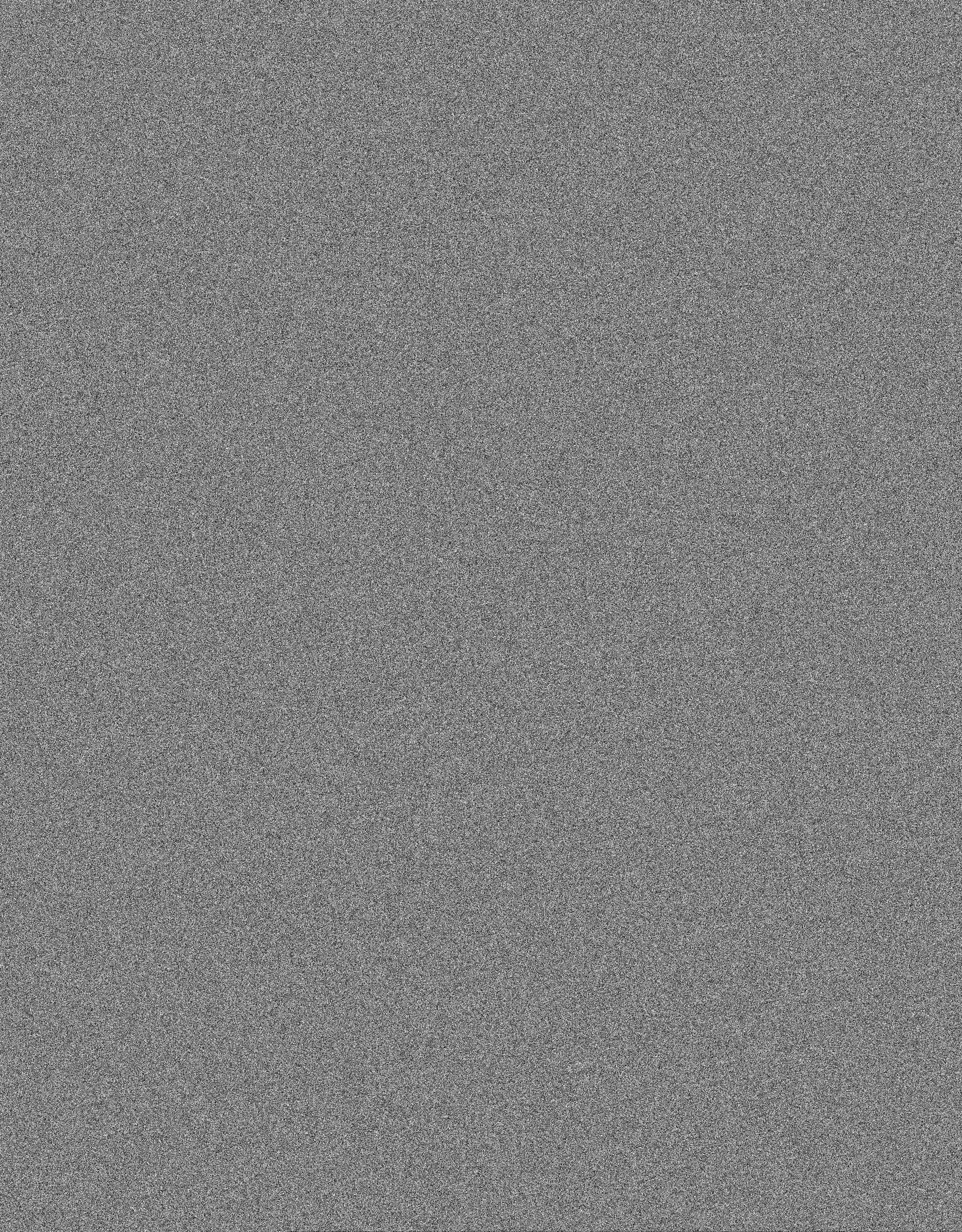}
  \caption{Extracted 20,431,465 bits, arranged row by row from left to right, and from top to bottom. A white pixel represents a "1", a black pixel marks a "0".}
  \label{fig:bits}
  \end{center}
\end{figure*}

\begin{figure*}
    \centering
    \includegraphics[angle=90,width=\columnwidth]{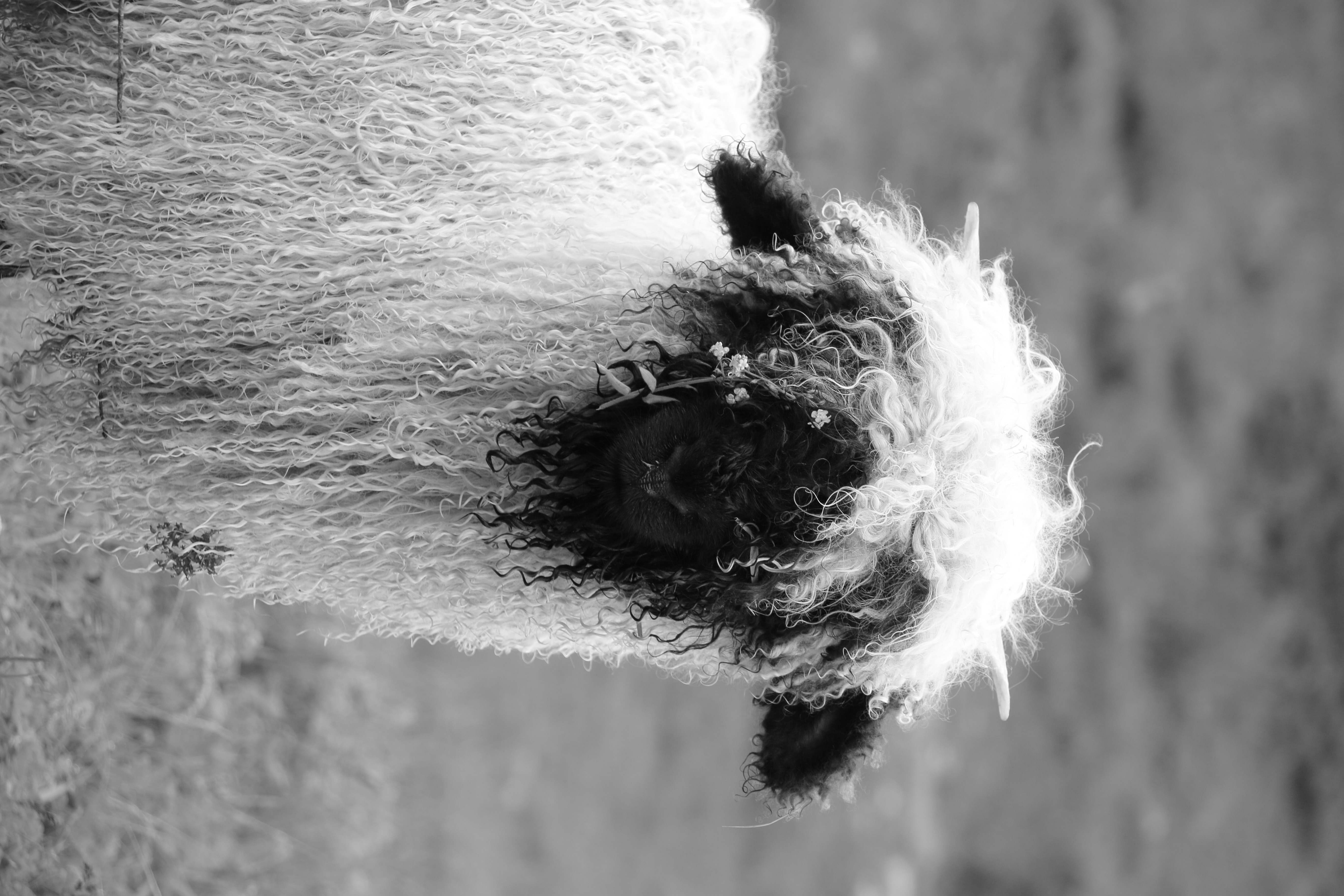}
    \includegraphics[angle=90,width=\columnwidth]{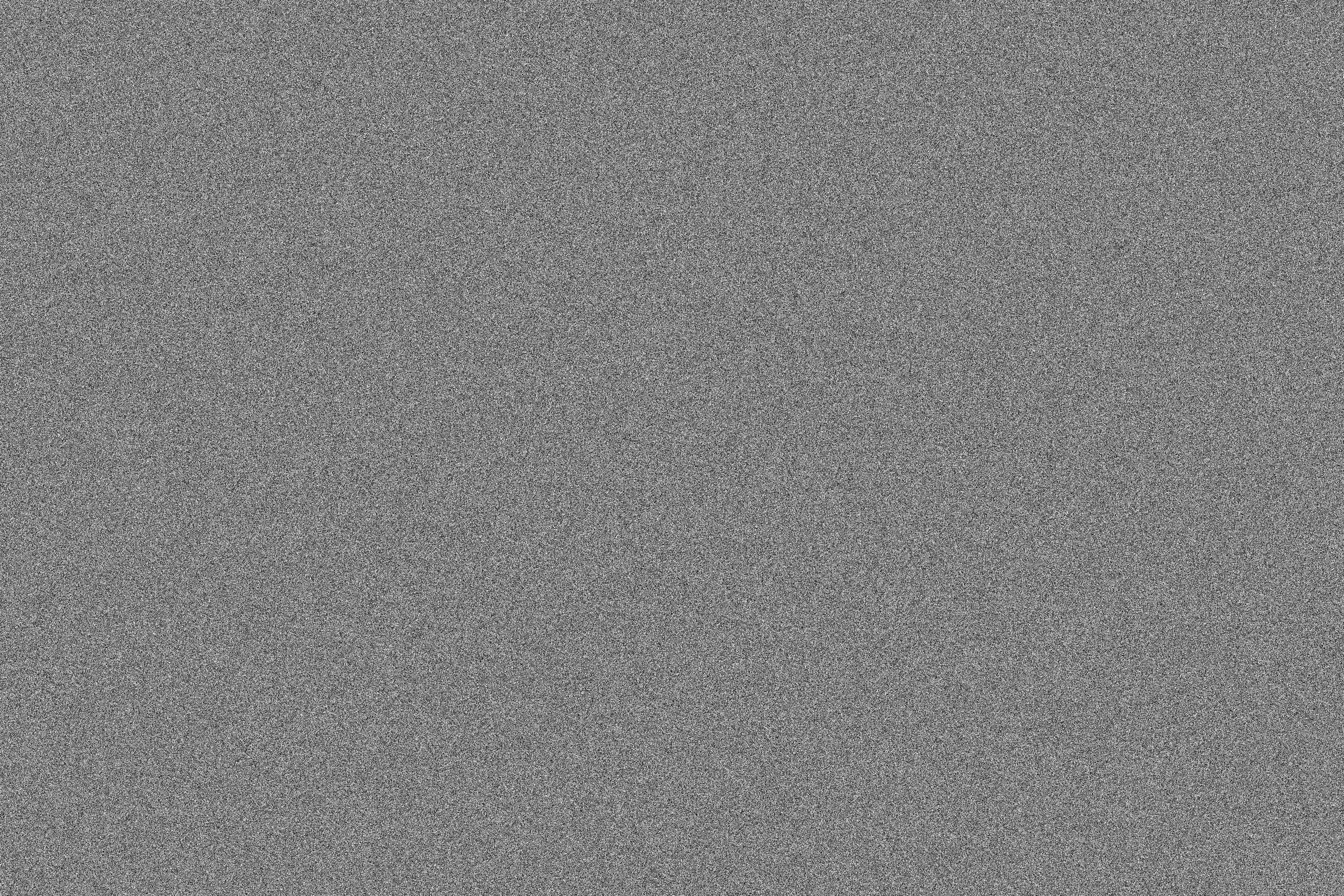}
    \caption{Sheep and one-time padded sheep.}
    \label{fig:sheep}
\end{figure*}

\clearpage

\bibliographystyle{apsrev4-1}
\bibliography{RefDB/basic_refs,RefDB/extra_refs,RefDB/DIRA_refs,RefDB/qudevRefDB_part}

\end{document}

%% file: figures/setup.tex
\newcommand{\mybox}[1]{
    \draw[fill=black,opacity=0.4,draw=white,thick] (-1+#1/2,-1+#1/2) -- (1+#1/2,-1+#1/2) -- (1+#1/2,1+#1/2) -- (-1+#1/2,1+#1/2) -- cycle;
    \draw[fill=black,opacity=0.6,draw=white,thick] (-1-#1/2,-1-#1/2) -- (1-#1/2,-1-#1/2) -- (1+#1/2,-1+#1/2) -- (-1+#1/2,-1+#1/2) -- cycle;
    \draw[fill=black,opacity=0.4,draw=white,thick] (1-#1/2,-1-#1/2) -- (1+#1/2, -1+#1/2) -- (1+#1/2,1+#1/2) -- (1-#1/2,1-#1/2) -- cycle;
    \draw[fill=black,opacity=0.4,draw=white,thick] (-1-#1/2,-1-#1/2) -- (-1+#1/2, -1+#1/2) -- (-1+#1/2,1+#1/2) -- (-1-#1/2,1-#1/2) -- cycle;
    \draw[fill=black,opacity=0.2,draw=white,thick] (-1-#1/2,1-#1/2) -- (1-#1/2,1-#1/2) -- (1+#1/2,1+#1/2) -- (-1+#1/2,1+#1/2) -- cycle;
    \draw[fill=black,opacity=0.4,draw=white,thick] (-1-#1/2,-1-#1/2) -- (1-#1/2,-1-#1/2) -- (1-#1/2,1-#1/2) -- (-1-#1/2,1-#1/2) -- cycle;
}

\begin{tikzpicture}[
    scale = 0.8,
    boxnode/.style={
        align=center,
        draw=black,
        minimum height=1.0cm,
        fill=myblue!60,
        rounded corners=0.1cm,
        blur shadow={shadow blur steps=15}
    },
]

    \begin{scope}[shift={(0, +2)},yscale=0.6]
        \mybox{0.5};
    \end{scope}

    \begin{scope}[shift={(0, -2)},yscale=0.6]
        \mybox{0.5};
    \end{scope}

    \node (X) at (0.0, +4.0) {};
    \draw[thick,->,>=stealth] ([yshift=-0.2cm] X.south) -- node[right=0.05cm,pos=0.6]{$\mathbf{X}$} (0.0, 2.8);
    \draw[thick,snake it] (0, -1) -- (0, 1);
    \begin{scope}[on background layer]
        \draw[fill=brown!15,draw=none,rounded corners=0.1cm] (-1.5, -3.5) -- (1.5, -3.5) -- (1.5, 3.5) -- (-1.5, 3.5) -- node[right=0.1cm, text=brown]{$\lambda$} cycle;
    \end{scope}

    \node (Y) at (0.0, -4.0) {};
    \draw[draw=none, fill=blue!25, rounded corners=0.1cm] (-1.5, -3.7) -- (1.5, -3.7) -- (1.5, -4.4) -- (-1.5, -4.4) -- cycle;
    \node at (0, -4.05) {0 1 1 1 1 0 1 $\ldots$};
    \draw[very thick, decorate,decoration={brace, raise=0.1cm, amplitude=0.15cm}] (1.5, -4.4) -- node[below=0.3cm] {SV source} (-1.5, -4.4);
    \draw[thick,->,>=stealth] ([yshift=0.2cm] Y.north) -- node[right=0.05cm,pos=0.6]{$\mathbf{Y}$} (0.0, -2.8);/

    \node[thick,draw=black,minimum height=4cm,minimum width=1.3cm,fill=gray!50] (Ext) at (4.55, 0) {$\mathrm{Ext}$};

    \draw[thick] (1.5, +2.0) -- node[above]{$\mathbf{A}$} (2.5, +2.0) -- (2.5, -2.0) -- node[below]{$\mathbf{B}$} (1.5, -2.0); 
    \draw[thick,->,>=stealth] (2.5, 0) -- node[above]{$\mathbf{AB}$} (Ext.west);

    \node (Z) at (Ext.north|-X.center) {};
    \draw[thick,->,>=stealth] ([yshift=-0.2cm] Z.south) -- node[right=0.05cm]{$\mathbf{Z}$} ([yshift=0.0cm] Ext.north);
    \draw[thick,->,>=stealth] (Ext.east) -- ([xshift=1.0cm] Ext.east); 

    \node[fill=violet!25, rounded corners=0.1cm, anchor=west] (BitsOut) at ([xshift=1.2cm] Ext.east) {1 1 0 1 1 $\ldots$};
    \draw[very thick, decorate, decoration={brace, raise=0.1cm, amplitude=0.15cm}] (BitsOut.north west) -- node[above=0.3cm]{$\mathbf{K}$} (BitsOut.north east);

    \begin{scope}[on background layer]
        \node (tmp) at (0, 4.4) {};
        \node (tmp2) at (0, 4.05) {};
        \draw[fill=red!25,draw=none,rounded corners=0.1cm] (-1.5, 3.7) -- (5.6, 3.7) -- (5.6, 4.4) -- (-1.5, 4.4) -- cycle;
            \node (Bits1) at (0, 4.05) {0 1 1 1 0 1 0 $\ldots$};
        \node[anchor=west] (Bits2) at (Ext.west|-tmp2) {0 1 1 1 $\ldots$};

        \draw[very thick, decorate, decoration={brace, raise=0.1cm, amplitude=0.15cm}] (-1.5, 4.4) -- node[above=0.3cm]{SV source} (tmp-|Bits2.east);
    \end{scope}

\end{tikzpicture}